\documentclass[manuscript,screen]{acmart}

\usepackage{latexsym}
\usepackage{amsmath}
\usepackage{amsthm}
\usepackage{booktabs}
\usepackage{enumitem}
\usepackage{graphicx}
\usepackage{color}

\usepackage{graphicx}
\usepackage{textcomp}
\usepackage{xcolor}
\usepackage{comment}
\usepackage{algorithmic}

\usepackage{multirow}
\usepackage{subfigure}
\usepackage{color}
\usepackage{verbatim}
\usepackage{booktabs}
\usepackage{makecell}
\usepackage{varwidth}
\usepackage{bbding}
\usepackage[linesnumbered,ruled]{algorithm2e}
\usepackage{caption}

\newcommand{\gao}[1]{\textcolor{blue}{#1}}

\AtBeginDocument{%
  \providecommand\BibTeX{{%
    \normalfont B\kern-0.5em{\scshape i\kern-0.25em b}\kern-0.8em\TeX}}}

\setcopyright{acmlicensed}
\copyrightyear{2024}
\acmYear{2024}
\acmDOI{XXXXXXX.XXXXXXX}

\begin{document}

\title{Quantifying and Defending against the Privacy Risk in Logit-based Federated Learning}

\author{Sheng Wan}
\email{swanae@connect.ust.hk}
\affiliation{%
  \institution{HKUST \& SUSTech}
  \city{Shenzhen}
  \country{China}
}

\author{Dashan Gao}
\email{dgaoaa@cse.ust.hk}
\affiliation{%
  \institution{HKUST \& SUSTech}
  \city{Shenzhen}
  \country{China}
}

\author{Hanlin Gu}
\affiliation{%
  \institution{Webank}
  \city{Shenzhen}
  \country{China}
  }
\email{}

\author{Lixin Fan}
\affiliation{%
  \institution{Webank}
  \city{Shenzhen}
  \country{China}
}
\email{}

\author{Daning Hu}
\affiliation{%
  \institution{SUSTech}
  \city{Shenzhen}
  \country{China}
}
\email{}

\author{Qiang Yang}
\affiliation{%
 \institution{PolyU}
 \city{Hong Kong}
 \country{China}
 }
\email{}

\renewcommand{\shortauthors}{Sheng et al.}

\begin{abstract}
  Federated learning (FL) aims to protect data privacy by collaboratively learning a model without sharing private data among clients. Unlike traditional parameter-based FL methods that exchange model weights or gradients during training, emerging logit-based FL approaches share model outputs (logits) on public data. This strategy promotes model heterogeneity, reduces communication overhead, and enhances clients' privacy. However, the potential privacy risks associated with these logit-based methods have been largely overlooked. This research presents the first theoretical and empirical analysis of a hidden privacy risk in logit-based FL methods – the risk that a semi-honest server (adversary) may learn clients' private models from logits. To quantify and address this threat, we develop the Adaptive Model Stealing Attack (AdaMSA) by leveraging historical logits during training. Notably, we observe that this inherent privacy risk persists even when public data is unrelated to private data, emphasizing the urgency to address privacy vulnerabilities in logit-based FL methods. Moreover, our theoretical analysis establishes the bounds of this privacy risk. We then propose a simple but effective defense strategy that perturbs the transmitted logits in the direction that minimizes the privacy risk while maximally preserving the training performance. We further provide a theoretical analysis of FedLP's privacy and utility properties. The experimental results validate our analysis and demonstrate the effectiveness of AdaMSA and our defense strategy.
\end{abstract}




\settopmatter{printacmref=false}

\maketitle

\section{Introduction}
In recent years, data privacy regulations such as the General Data Protection Regulation have largely restricted the collection of annotated data on individuals for centralized training. Federated Learning (FL)~\cite{mcmahan2016communication} offers a promising approach that enables different clients to collaboratively train their models by sharing local model parameters or gradients without exchanging their respective raw data. However, recent studies~\cite{zhu2019deep,geiping2020inverting} have introduced gradient-based attacks, demonstrating the potential extraction of private training data from shared gradients or parameters. This revelation poses a substantial privacy leakage concern.

Another line of FL studies~\cite{gong2021ensemble,gong2022preserving,jeong2018communication,li2019fedmd} adopt knowledge distillation~\cite{hinton2015distilling} to \textbf{exchange model outputs (i.e., logits) instead of model weights or gradients during training} to reduce communication overhead and enable model heterogeneity. To differentiate this from the traditional federated learning approaches, which transmit parameters, we refer to this method as \textbf{logit-based federated learning}. 
To preserve clients' privacy, these logit-based FL methods distill on public data to transfer knowledge and model parameters are stored locally during training, as depicted in Figure \ref{fig1}. Moreover, such public data can be unlabeled and insensitive that is sampled from other domains~\cite{gong2022preserving}. 

A natural question arises: \textit{Can the logit-sharing scheme ensure the privacy of each participant?} Intuitively, restricting clients to share only the logits—i.e., the outputs of the final layer of the model—instead of sharing the entire model can be considered a form of data minimization. This approach appears to reduce the potential disclosure of private information within the shared updates sent to the adversary. However, we argue that the transmitted logits still pose a risk of privacy leakage for the private model. As the model remains hidden on local machines, it could still be exploited by adversaries as a prerequisite for gradient-based attacks \cite{zhu2019deep,geiping2020inverting} or membership inference attacks \cite{7958568}. Therefore, ensuring model privacy should be a fundamental requirement for logit-based federated learning.

\begin{centering}
\begin{figure}[t] 
\includegraphics[width=\linewidth]{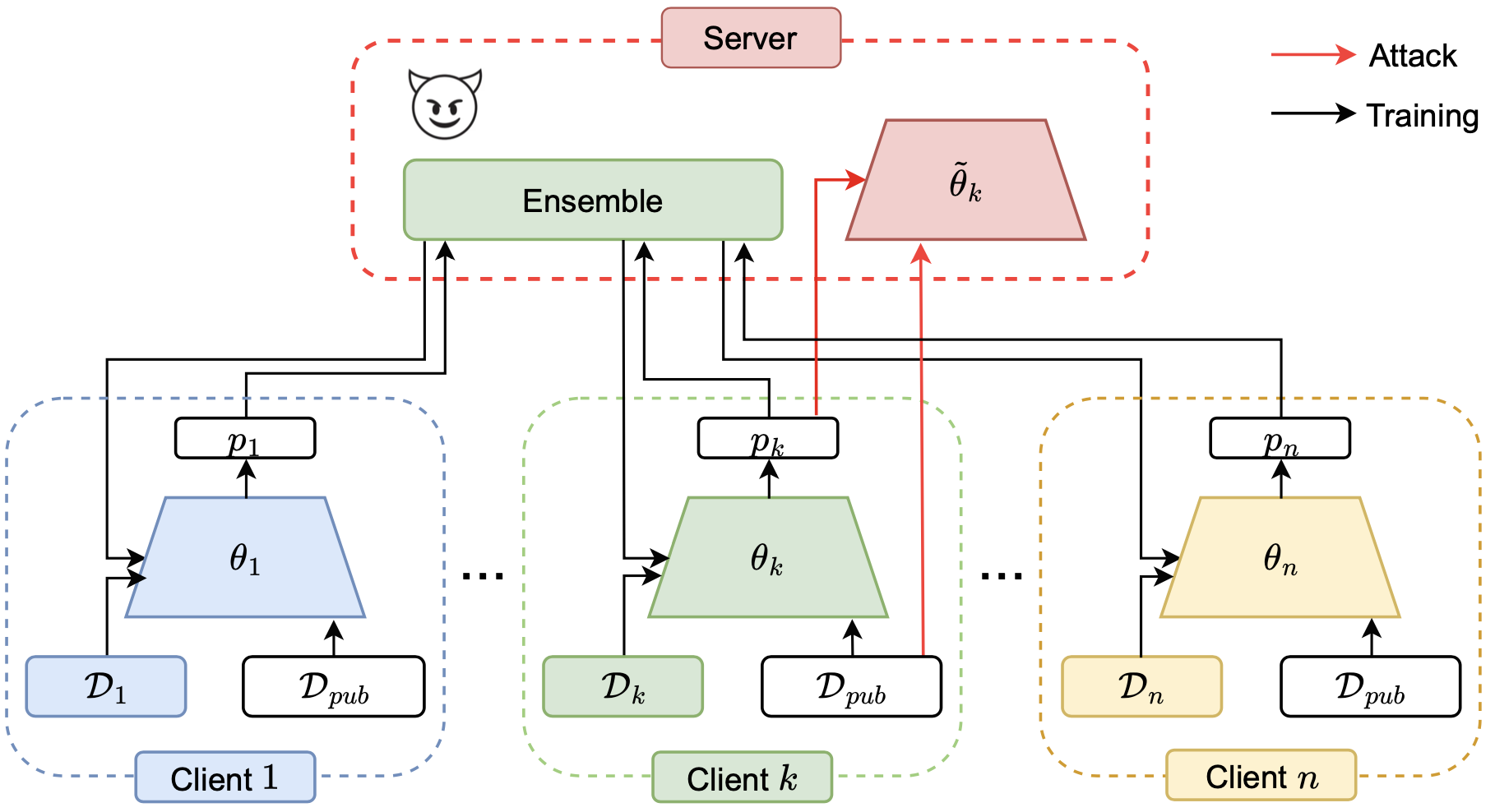}
\caption{Illustration of logit-based FL and our attack setting. Private data and models are maintained locally, with only local predicted logits being transmitted to the server. The server is semi-honest and aims to infer client $k's$ private model via its logits $p_k$ on an unlabeled public dataset $\mathcal{D}_{pub}$ during training.} \label{fig1}
\end{figure}
\end{centering}

To our knowledge, we are the first to provide a theoretical and empirical analysis of a hidden privacy risk in logit-based FL that the semi-honest server intends to infer clients' private models without knowing local model architecture or data distribution. To quantify the impacts of the privacy risk, we develop an effective attack dubbed as Adaptive Model Stealing Attack (AdaMSA), which adaptively steals the private model by approximating its intermediate training states from previous iterations. In each iteration, the semi-honest server compels the attacking model to approximate the current state of the victim model by minimizing the distance between the output of the attacking model and a target logit. We propose combining the observed historical logits of the victim model to capture diverse aspects of private information. This approach enables us to obtain a more informative target logit, as empirically validated in Section~\ref{attack_eval}. Moreover, we provide a theoretical analysis establishing the bounds of this privacy risk in logit-based FL in Section~\ref{section3-3}.

To protect model privacy in logit-based FL, we propose a simple but effective perturbation-based defense strategy named Federated Logit Perturbation (FedLP). The key idea of our strategy is to perturb the logit in the direction that optimally obstructs the adversary while maximally preserving the model performance. As a result, our defense achieves a better trade-off compared to prior methods.

\begin{centering}
\begin{table*}[t] 
\captionof{table}{Existing privacy attacks against FL.}\label{table5}
\centering
\resizebox{\textwidth}{!}{
\begin{tabular}{ccccccc}
\toprule
Attack & Adversary & Adversary's Goal & \multicolumn{3}{c}{Adversary’s Knowledge} & Applicable for \\
\cline{4-6}
& & & Model &  Gradient & Logit & Logit-based FL\\
\hline
CPA \cite{nasr2019comprehensive} &  \makecell[c]{Semi-honest/ \\ Malicious Client} & Infer Membership & $\checkmark$ & $\checkmark$ & $\times$ & $\times$\\
mGAN \cite{wang2019beyond} &  Malicious Client & Infer Class Representative & $\checkmark$ & $\checkmark$ & $\times$ & $\times$\\
UFL \cite{melis2019exploiting} &  \makecell[c]{Semi-honest/\\ Malicious Client} & Infer Property & $\checkmark$ & $\checkmark$ & $\times$ & $\times$\\
DLG \cite{zhu2019deep} & Semi-honest Server & Infer Training Data & $\checkmark$ & $\checkmark$ & $\checkmark$ & $\times$\\
InvertGrad \cite{geiping2020inverting} &  Semi-honest Server & Infer Training Data & $\checkmark$ & $\checkmark$ & $\checkmark$ & $\times$\\
AdaMSA (ours) & Semi-honest Server & Infer Model & $\times$ & $\times$ & $\checkmark$ & $\checkmark$\\
\bottomrule
\end{tabular}
}
\end{table*}
\end{centering}

We empirically evaluate our proposed attack and defense approach in three experimental settings, Close-world, Open-world-CF and Open-world-TI (see Section \ref{settings} for details). Our empirical findings validate our theoretical analysis and demonstrate the effectiveness of AdaMSA, revealing its success even in scenarios where public data lacks relevance to private data. Additionally, our proposed FedLP achieves a better utility and privacy trade-off compared to state-of-the-art baselines. Our observation underscores the importance of this hidden privacy risk and highlights the urgency to address privacy vulnerabilities in logit-based FL methods. We hope our research can serve as a valuable tool, aiding researchers in evaluating privacy risks within logit-based FL and fostering the development of privacy-preserving FL methodologies.

Our key contributions are summarized as follows:
\begin{itemize}
    \item[$\bullet$] To the best of our knowledge, we provide the first theoretical and empirical analysis of a hidden privacy risk in logit-based FL that the semi-honest server can infer clients' private models according to logits. 
    \item[$\bullet$] To quantify the privacy risk, we develop an effective model stealing attack named AdaMSA, which steals private models by leveraging historical logits during training. Moreover, we provide a theoretical bound for the privacy risk in logit-based FL.
    \item[$\bullet$] To prevent the privacy risk, we develop a simple but effective defense named FedLP by perturbing the transmitted logits in the direction that minimizes the privacy risk while maximally preserving the model performance. We further provide theoretical analyses of its privacy and utility properties.
    \item[$\bullet$] We conduct empirical evaluations of our designed attack across three experimental settings: Close-world, Open-world-CF, and Open-world-TI. The results not only validate our analysis but also demonstrate that AdaMSA can achieve up to a 4.00 percentage-point improvement and our defense can achieve a better utility and privacy trade-off compared to the state-of-the-arts.
\end{itemize}

\section{Related Work}
\subsection{Privacy Risk in Federated Learning}
Federated learning \cite{kairouz2019advances} allows multiple clients to collaboratively train a global model while keeping training data locally. Typical FL algorithms \cite{mcmahan2016communication,karimireddy2019scaffold} are parameter-based FL that shares local model parameters or gradients and aggregate local models in the server. Logit-based FL \cite{gong2021ensemble,gong2022preserving,jeong2018communication,li2019fedmd} adopt knowledge distillation \cite{hinton2015distilling} to transmit model outputs (i.e., logits) instead of model weights or gradients during training to reduce communication overhead, enable model to be heterogeneous and preserve clients' privacy. 

Previous studies have extensively investigated the privacy risks of sharing model parameters or gradients in FL, addressing concerns like class representatives leakage \cite{wang2019beyond}, membership leakage \cite{nasr2019comprehensive}, property leakage \cite{melis2019exploiting} and training input leakage \cite{geiping2020inverting,zhu2019deep}. However, these efforts primarily concentrate on gradient-based attacks under white-box assumptions, as detailed in Table \ref{table5}. In other words, they have strong assumptions that the adversary knows the local model architecture and detailed training information such as gradients. In this work, we focus on logit-based FL, where model parameters or gradients are stored in clients’ local machines. 
A recent empirical study \cite{takahashi2023breaching} delves into breaching the logit-based FL framework. While their emphasis lies in reconstructing class representations of private data and is contingent on the relevance of public data to private data, our objective is to steal the functionality of the private model. Importantly, our approach can succeed even in scenarios where public data lacks relevance to private data.
To our knowledge, this research represents the first investigation providing both theoretical and empirical analyses on the privacy risk associated with logit-based FL.


\subsection{Model Stealing Attack}
Model stealing attacks \cite{orekondy2019knockoff,papernot2017practical,tramer2016stealing} have demonstrated the ability to steal a deployed machine learning model in a black-box manner through limited query access and carefully calibrated proxy dataset. These attacks happen in the inference stage and aim to reduce the number of queries or eliminate the need of proxy dataset. However, in logit-based FL, the attack happens in the training stage, where the adversary can neither arbitrarily select the query dataset nor access to the private models or private dataset distribution. Instead, the adversary only observes the intermediate information (i.e., transmitted logits) from the victim during training. Based on this observation, we propose AdaMSA that leverages historical training information to obtain more informative target logits and therefore improve the attack performance.

\subsection{Privacy Protection Strategy in Logit-based FL}
Researchers have proposed some strategies \citep{li2019fedmd,gong2022preserving,sattler2021fedaux} to prevent the potential privacy leakage in logit-based FL. Specifically, Li et al. \citep{li2019fedmd} proposed to distill on a public dataset instead of private data to transfer predicted vectors. Gong et al. \citep{gong2022preserving} further relaxed the public data to be unlabeled and insensitive data sampled from other domains to preserve data privacy. Moreover, Sattler et al. \citep{sattler2021fedaux} and Gong et al. \citep{gong2022preserving} adopted differential privacy (DP) to protect the transmitted logits. However, these papers fail to quantify the privacy risk inside logit-based FL and their defense strategies incur a significant loss in accuracy. In contrast, we first identify and quantify the privacy risk. Then we design a simple but effective perturbation strategy against our revealed privacy risk, which perturbs the logit in the direction that maximally misleads the adversary while preserving training performance. Therefore, it can achieve a better utility and privacy trade-off.

\section{Quantifying the Privacy Risk} \label{section-3}
In this section, we begin by providing the problem setup and threat model. Then we propose an attack to quantify the privacy risk in logit-based FL and elaborate our proposed attack in detail. Lastly, we give a theoretical analysis on the bound of this privacy risk.

\subsection{Problem Setup and Threat Model}

\textbf{Problem Setup:} As shown in Figure 1, there are $n$ clients and a central server. Each client has a private labeled dataset $\{\mathcal{D}_i\}_{i=1}^n$ and some unlabeled public data $\mathcal{D}_{pub}$. The server coordinates the training process by aggregating the logits predicted locally by the clients. 
Specifically, the model $f(x, \theta)$ takes input data $x \in \mathcal{D}_{pub}$ and outputs logits $p = f(x, \theta)$, where logits represent the unnormalized prediction scores for each class.
The locally predicted logits $\{p_i\}_{i=1}^n$ are aggregated by the server to compute an ensemble logit for $\mathcal{D}_{pub}$. Then, clients train their local models $\{\theta_i\}_{i=1}^n$ under the supervision of labels on $\{\mathcal{D}_i\}_{i=1}^n$ and the ensemble logits on $\mathcal{D}_{pub}$. In the subsequent discussion, we use logits $p$ to represent the outputs of the model $f(x, \theta)$ for simplicity.

\textbf{Threat Model:} We assume that the server (i.e., adversary) is \textit{semi-honest}, i.e., it completes the learning task as required but is curious about the clients’ local models. 
All clients are assumed honest. 
The adversary does not know the private data distribution or the victim model, including its parameters, hyperparameters, or architecture. Moreover, the adversary does not have the right to select public data. The adversary only knows that: 1) the unlabeled public dataset $\mathcal{D}_{pub}$; 2) the transmitted victim’s logits $\{p^t\}_{t=1}^T$ on the public data during the training process with $T$ iterations in total. 
The adversary aims to approximate the functionality of the victim’s model $f(x, \theta)$ by training an attacking model $f(x, \tilde{\theta})$ to achieve high accuracy. Specifically, the adversary seeks to maximize the classification accuracy (Acc) of the attacking model on the victim $k$’s private dataset $\mathcal{D}_k$:
\begin{equation*}
    \max_{\tilde{\theta}} \mathbb{E}_{x \sim \mathcal{D}_k} \text{Acc}\left(f(x, \tilde{\theta})\right).
\end{equation*}

We summarize the threat model in Table \ref{table1}. In the following discussion, we assume that the adversary is interested in client $k$'s model, denoted as $\theta$ for simplicity.

\begin{centering}
\begin{table}[h] \centering
\captionof{table}{Threat model.}\label{table1}
\begin{tabular}{cccc}
\toprule
Threat Model & Adversary & Attack Target & Adversary’s Knowledge \\
\hline
Semi-honest & Server & Private model $\theta$ & \makecell[c]{$\mathcal{D}_{pub}$ and logits of $\theta$ on $\mathcal{D}_{pub}$. }\\
\bottomrule
\end{tabular}
\end{table}
\end{centering}

\textbf{Notations:} Table~\ref{tab_notation} provides a summary of the notations employed throughout this paper.

\begin{table}[t]
    \centering
    \caption{Table of Notations}
    \label{tab_notation}
    \begin{tabular}{ll}
        \toprule
        \textbf{Symbol} & \textbf{Description} \\
        \midrule
        $\mathcal{D}_{pub}$ & Public dataset used for training and evaluation. \\
        $\mathcal{D}_k$ & Private dataset of the victim $k$. \\
        $x_i$ & A single data sample from the public dataset $\mathcal{D}_{pub}$. \\
        $f(x, \theta)$ & Victim model with parameters $\theta$. \\
        $f(x, \tilde{\theta})$ & Attacking model with parameters $\tilde{\theta}$. \\
        $p_i^t$ & Victim logits of $x_i$ predicted by the victim model at iteration $t$. \\
        $\hat{p}_i^T$ & Target logits for $x_i$ combining iterations $T-T_0$ to $T$. \\
        $\tilde{p}_i^T$ & Attacking logits of $x_i$ predicted by the attacking model at iteration $T$. \\        $T$ & Maximum number of training iterations. \\
        $T_0$ & Threshold controlling the range of past predictions. \\
        $w_t$ & Weight assigned to logits from iteration $t$. \\
        $Z_{norm}$ & Normalization factor ensuring $\sum_{t=T-T_0}^T w_t = 1$. \\
        $\eta$ & Learning rate used for gradient descent. \\
        $\mathcal{L}_{CE}$ & Cross-entropy loss function. \\
        \bottomrule
    \end{tabular}
\end{table}

\subsection{Adaptive Model Stealing Attack} \label{section3-2}
To quantify the privacy risk inherent in logit-based FL, we propose an Adaptive Model Stealing Attack, denoted as AdaMSA. 
AdaMSA adaptively steals the victim's private model by approximating its intermediate training states from prior iterations during training. Specifically, in each iteration, the server compels the attacking model to approximate the victim model by mimicking a designated target logit on public data. The crucial challenge here lies in the design of the target logit, as a more informative target logit can enhance the supervision for the attacking model. 

Given the historical logits during training, we formulate the design of the target logit with two key considerations: 1) Historical predictions may contain valuable information, offering diverse perspectives of the data to enhance the supervision for the attacking model; 2) Predictions in the early rounds may not be sufficiently trained to provide informative supervision. 
Accordingly, we define the target logit $\hat{p}^T$ as
\begin{equation*}
\hat{p}^T = \sum_{t=T-T_0}^T \frac{t}{Z_{norm}} \cdot p^t,
\end{equation*}
\noindent where $T_0$ is the threshold governing the extent of past predictions to be considered, and $Z_{norm} = \sum_{k=T-T_0}^T k$ is a normalization factor that ensures the weights $w_t = \frac{t}{Z_{norm}}$ sum to 1, effectively balancing the contributions of logits from different iterations.
To emphasize predictions closer to the current iteration $T$, we let $w_t$ increase linearly with the iteration number $t$.

To train the attacking model $\tilde{\theta}$, we formulate the cross-entropy loss $\mathcal{L}_{CE}$ as the discrepancy between the target logit $\hat{p}^T$ and the prediction of the attacking model (the attacking logit) $\Tilde{p}^T$ on $\mathcal{D}_{pub}$, defined as follows:
\begin{equation*}
    {\rm min}_{\tilde{\theta}} \ \mathbb{E}_{x \sim \mathcal{D}_{pub}} [\mathcal{L}_{CE}(\hat{p}^T,\Tilde{p}^T)].
\end{equation*}
Then we employ gradient descent to address this optimization problem:
\begin{equation*}
    \tilde{\theta} = \tilde{\theta} - \eta \nabla_{\tilde{\theta}} \tilde{\mathcal{L}}_{CE}(\hat{p}^T,\Tilde{p}^T).
\end{equation*}

The details of AdaMSA in one training iteration are given in Algorithm 1. By repeating this process, the semi-honest server is able to steal any desired intermediate private model throughout training, including the ultimately well-trained private model of the victim. Subsequently, we can quantify the privacy risk $\mathcal{R}$ in logit-based FL by evaluating the accuracy of the obtained attacking model $f(x,\tilde{\theta})$ on the victim's private test dataset $\mathcal{D}_{k}$ as follows:
\begin{equation*}
    \mathcal{R} = \mathbb{E}_{x \sim \mathcal{D}_{k}} \text{Acc}\left(f(x, \tilde{\theta})\right).
\end{equation*}

A higher accuracy indicates a greater degree of privacy leakage from the victim.

Here we provide a brief analysis of the computational complexity associated with our proposed attack. Assume the communication round is $M$. Then the computational complexity of AdaMSA, as described in Algorithm 1, is expressed as $O(M|\mathcal{D}_{pub}|)$. In the context of logit-based Federated Learning, the communication cost, also proportionate to $M|\mathcal{D}_{pub}|$, is expected to be bounded due to the inherent communication bottleneck. Consequently, our attack will not be expected to impose a significant burden on the server. 

\begin{algorithm}[t] 
\caption{AdaMSA in one iteration}
\begin{algorithmic}[1]\label{alg_AdaMSA_optimized} 
\REQUIRE Public dataset $\mathcal{D}_{pub} = \{x_i\}_{i=1}^N$, victim's logits $\{p_i^t\}_{i=1, t=1}^{N, T}$ predicted on $\mathcal{D}_{pub}$ in the past $T$ iterations, attacking model with parameters $\tilde{\theta}$, maximum number of iterations $T$, and learning rate $\eta$.
\ENSURE Optimized attacking model parameters $\tilde{\theta}$.
\WHILE{not converged} 
    \STATE Randomly shuffle the public dataset $\mathcal{D}_{pub}$
    \FOR{$x_i \in \mathcal{D}_{pub}$}
        \STATE Compute the ensemble of historical logits for $x_i$: $ \hat{p}_i^T \leftarrow \sum_{t=T-T_0}^T \frac{t}{Z_{norm}} \cdot p_i^t$
        \STATE Compute the predicted logits from the attacking model: $ \Tilde{p}_i^T \leftarrow f(x_i, \tilde{\theta}) $
        \STATE Update the attacking model $\tilde{\theta}$ via cross-entropy gradient: $ \tilde{\theta} \leftarrow \tilde{\theta} - \eta \nabla_{\tilde{\theta}} \mathcal{L}_{CE}(\hat{p}_i^T, \Tilde{p}_i^T) $
    \ENDFOR
\ENDWHILE
\STATE \textbf{return} $\tilde{\theta}$
\end{algorithmic}
\end{algorithm}

\subsection{Privacy Risk Analysis} \label{section3-3}

Logit-based FL methods transfer knowledge through the transmitted logits on public data during training. As shown in Figure \ref{fig2}, we discern that the distance measured by the correlation between the private dataset and the public dataset plays a crucial role in determining the performance of logit-based FL methods. To have a deeper understanding of the inherent cause of privacy risk in this logit-sharing scheme, we start with quantifying the distance between the private dataset and the public dataset. 

\begin{figure}[h]\centering 
\subfigure[Open-world-CF]{
\includegraphics[width=0.4\linewidth]{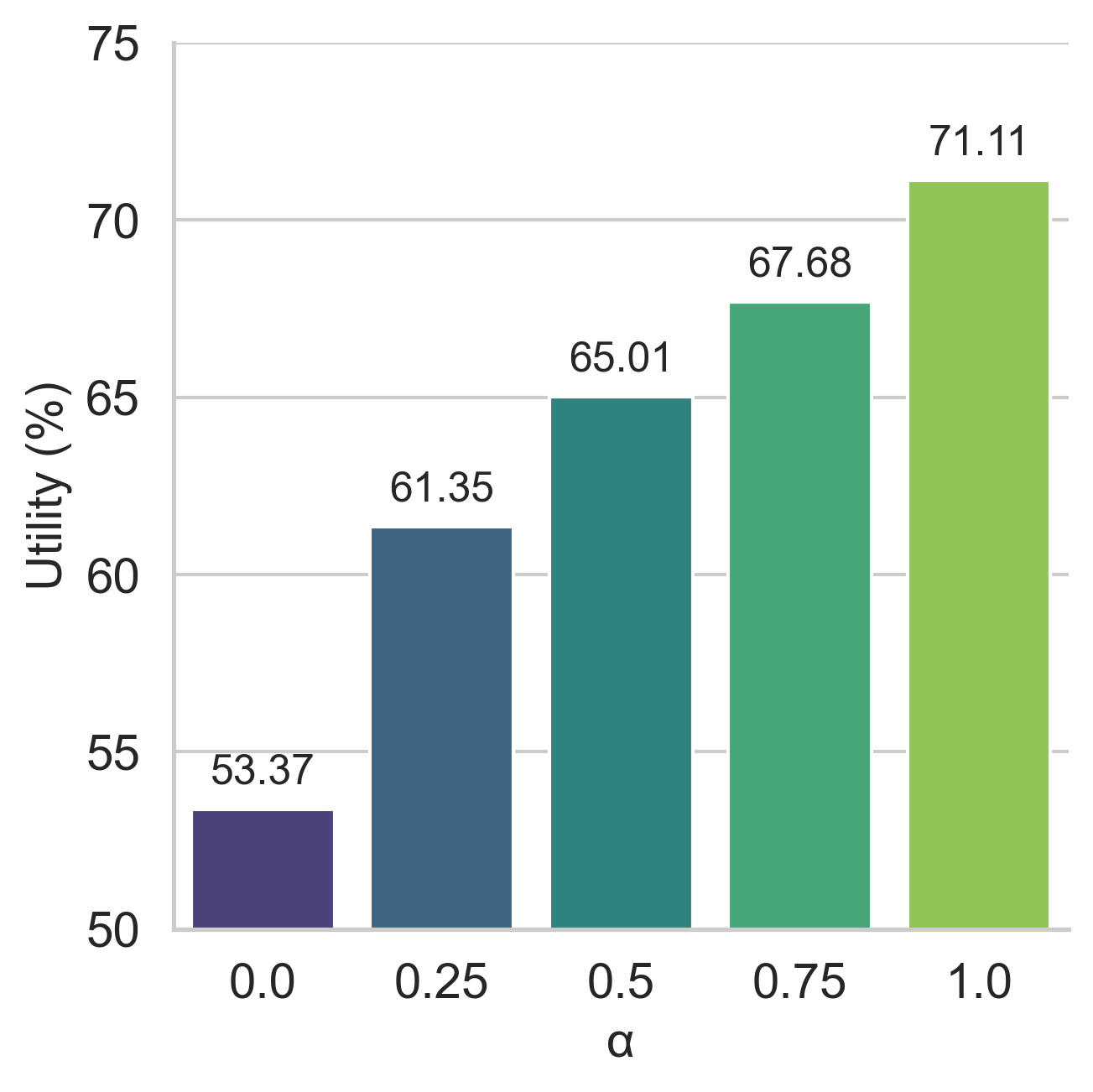}
}
\subfigure[Open-world-TI]{
\includegraphics[width=0.4\linewidth]{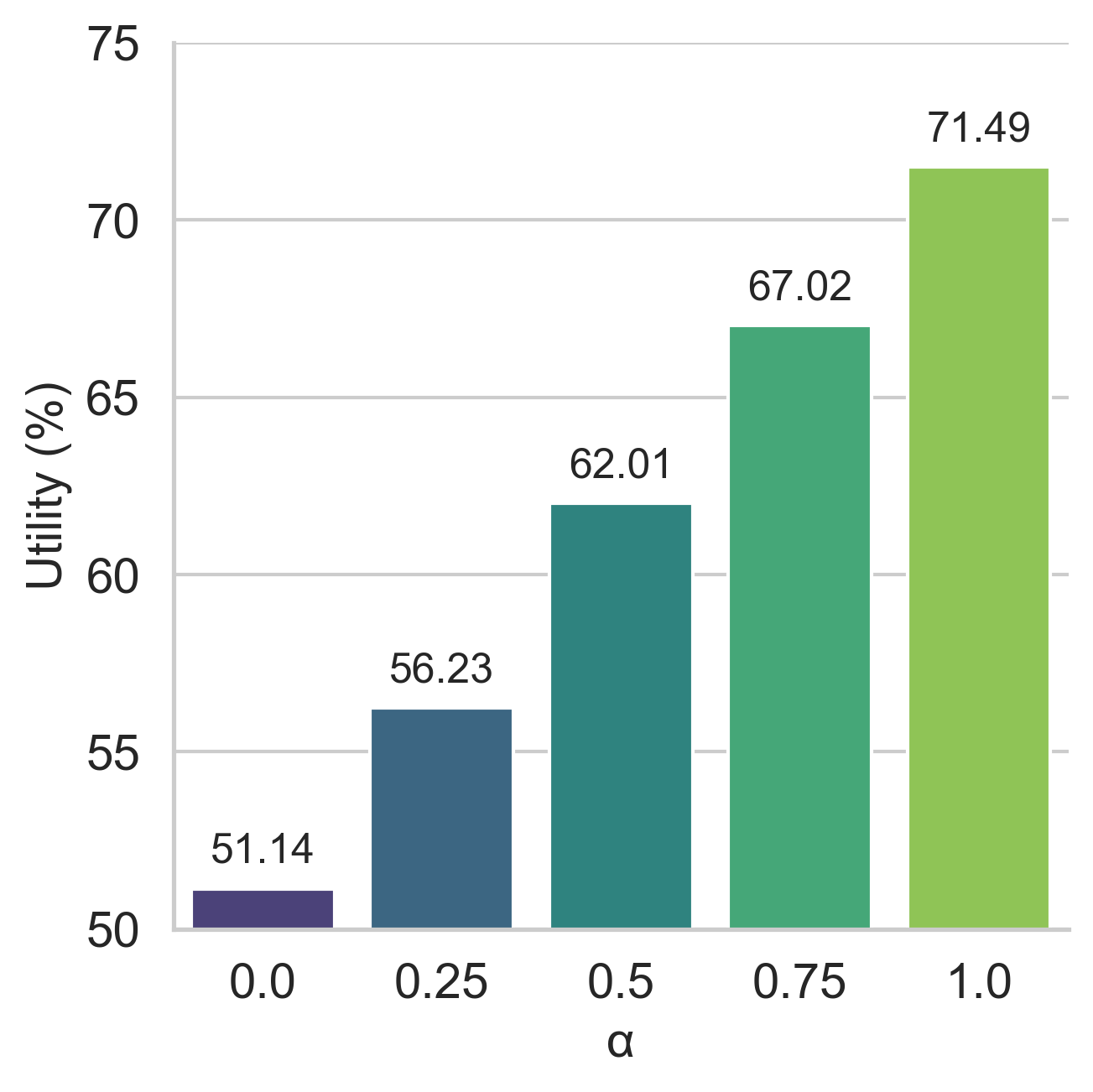}
}
\caption{ The impact of the distance between the private dataset and the public dataset on the utility of logit-based FL in Open-world-CF and Open-world-TI settings (See Section \ref{settings} for details). The distance is controlled by the weighting parameter $\alpha$.} 
\label{fig2}
\end{figure}

Previous research has not extensively explored systematic methods for adjusting the distance between distinct datasets. To address this issue, we propose to \emph{construct a mixed dataset, which is then employed as the public dataset in the standard logit-based FL setting}. This design equips us with the ability to systematically manipulate the distance between the mix and private datasets.

\begin{figure}[h]
    \centering
    \includegraphics[width=0.8\linewidth]{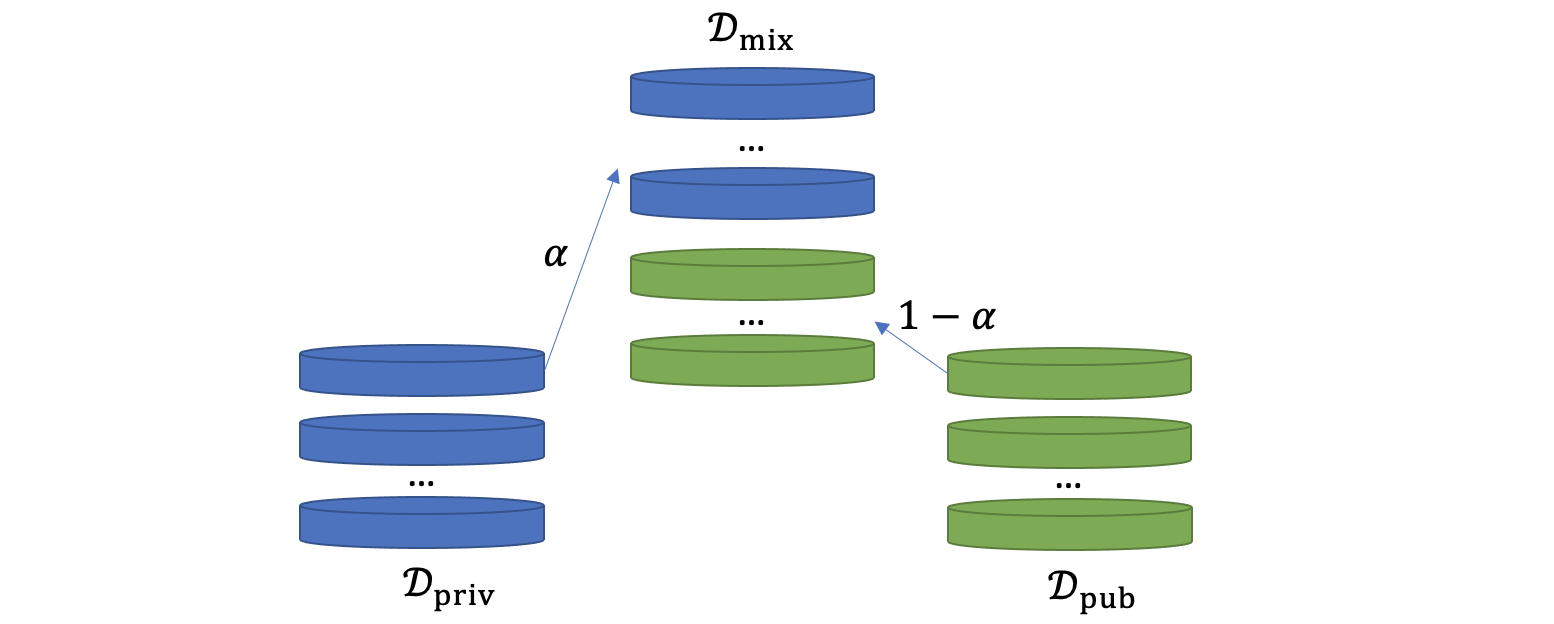}
    \caption{Illustration of constructing the mixed dataset $\mathcal{D}_{mix}$.}
    \label{fig_pub}
\end{figure}

Consider a simple case that we have $n$ private datasets sampled from the same distribution $\mathcal{D}_{priv}$ and an unlabeled public dataset from another domain sampled from independent distribution $\mathcal{D}_{pub}$. As shown in Figure \ref{fig_pub}, the mixed dataset is constructed through $\mathcal{D}_{mix} = (S_1, S_2)$, where $S_1$ consists of $\alpha |\mathcal{D}_{mix}|$ instances sampled independently from $\mathcal{D}_{priv}$ and $S_2$ consists of $(1-\alpha)|\mathcal{D}_{mix}|$ instances sampled independently from $\mathcal{D}_{pub}$. Through varying the weighting parameter $\alpha$, we can thereby quantify and control the distance between private and mixed datasets. For instance, as $\alpha$ tends to 1, the mixed dataset approaches to the private dataset, and vice versa.

As previously illustrated in Section \ref{section3-2}, the privacy risk is measured by the performance of AdaMSA on the victim's private test dataset. In consequence, we can establish the bound of the privacy risk via the performance bound of AdaMSA on the victim's private test dataset based on the prior art in domain adaptation \cite{blitzer2007learning}. 

Denote the empirical risk of model $\theta$ on mixed dataset as
\begin{equation} \label{Eq1}
    \epsilon_{\mathcal{D}_{mix}}(\theta,f_p) = \mathbb{E}_{x \sim \mathcal{D}_{mix}}[|\theta(x)-f_p(x)|],
\end{equation}
which measures the probability according to the distribution $\mathcal{D}$ that $\theta$ disagrees with the ground truth label $f_p$. For simplicity, we abbreviate $\epsilon(\theta,f_p)$ as $\epsilon(\theta)$. 

According to the definition provided in Equation \ref{Eq1}, we can express the empirical risk of the mixed dataset in terms of the empirical risk of the public and private datasets, as demonstrated in Theorem \ref{theorem1}.

\begin{theorem} \label{theorem1}
    Given a mixed dataset $\mathcal{D}_{mix} = (S_1, S_2)$, where $S_1$ consists of $\alpha |\mathcal{D}_{mix}|$ instances sampled independently from $\mathcal{D}_{priv}$, $S_2$ consists of $(1-\alpha)|\mathcal{D}_{mix}|$ instances sampled independently from $\mathcal{D}_{pub}$, its empirical risk can be written as
    \begin{equation*}
    \epsilon_{\mathcal{D}_{mix}}(\theta) = \alpha \epsilon_{\mathcal{D}_{priv}}(\theta) + (1-\alpha) \epsilon_{\mathcal{D}_{pub}}(\theta).
\end{equation*}
\end{theorem}

\textit{Proof.} According to the definition given in Equation (\ref{Eq1}), we have
\begin{align*}
        &\epsilon_{\mathcal{D}_{mix}}(\theta,f_p)\notag \\
        = & \mathbb{E}_{x \sim \mathcal{D}_{mix}}[|\theta(x)-f_p(x)|] \\
        = & \frac{1}{|\mathcal{D}_{mix}|} \sum\limits_{x_i \in \mathcal{D}_{mix}} |\theta(x_i)-f_p(x_i)| \\
        = & \frac{1}{|\mathcal{D}_{mix}|} [\sum\limits_{x_i \in \mathcal{D}_{priv}}|\theta(x_i)\!-\!f_p(x_i)|+ \sum\limits_{x_j \in \mathcal{D}_{pub}}|\theta(x_j)\!-\!f_p(x_j)|] \\
        = & \alpha \cdot \frac{1}{|\alpha \mathcal{D}_{mix}|} \sum\limits_{x_i \in \mathcal{D}_{priv}} |\theta(x_i)-f_p(x_i)| + \\
        & (1-\alpha) \cdot \frac{1}{|(1-\alpha) \mathcal{D}_{mix}|} \sum\limits_{x_j \in \mathcal{D}_{pub}} |\theta(x_j)-f_p(x_j)|\notag\\
        = & \alpha \epsilon_{\mathcal{D}_{priv}}(\theta) + (1-\alpha) \epsilon_{\mathcal{D}_{pub}}(\theta).
\end{align*}

Then we introduce some definitions for our latter analysis:
\begin{definition}[H-Divergence between Distributions]
Given a domain $\mathcal{X}$ with $\mathcal{D}$ and $\mathcal{D'}$ probability distributions over $\mathcal{X}$, let $\mathcal{H}$ be a hypothesis class on $\mathcal{X}$ and $\mathcal{A}_\mathcal{H}$ be the set of subsets of $\mathcal{X}$ that supports the hypothesis in $\mathcal{H}$. The H-divergence between $\mathcal{D}$ and $\mathcal{D'}$ is defined as: $d_H(\mathcal{D},\mathcal{D'})= 2 sup_{A\in A_H}|Pr_{\mathcal{D}}(A) - Pr_{\mathcal{D'}} (A)|$.

\end{definition}

\begin{definition}[Symmetric Difference Hypothesis Space]
For a hypothesis space $\mathcal{H}$, the symmetric difference hypothesis space $\mathcal{H}\Delta\mathcal{H}$ is defined as $\mathcal{H}\Delta\mathcal{H} = \{ h(x)\bigoplus h' (x)|h, h' \in \mathcal{H}\}$, where $\bigoplus$ represents the XOR operation.
\end{definition}

Subsequently, we can derive the bound of the difference between the empirical risk of $\theta$ on the mixed dataset $\mathcal{D}_{mix}$ and the private dataset $\mathcal{D}_{priv}$ and present our main theorem as follows:
\begin{theorem}
\label{theorem2}
Let $H$ be a hypothesis space of VC-dimension $d$ and $h$ be a hypothesis in class $H$. Let $\mathcal{D}_{mix}$ be the mixed dataset as illustrated in Fig. 3 and $\mathcal{D}_{priv}$ be the private dataset. Let $d_{\mathcal{H}\Delta\mathcal{H}}$ be the empirical distance induced by the symmetric difference hypothesis space. Then we can bound the difference between the empirical risks of the mixed dataset and the private dataset by
\begin{equation*}
    |\epsilon_{\mathcal{D}_{mix}}(h) - \epsilon_{\mathcal{D}_{priv}}(h) | \leq (1-\alpha)(\frac{1}{2}d_{\mathcal{H}\Delta\mathcal{H}}(\mathcal{D}_{priv}, \mathcal{D}_{pub}) +\lambda),
\end{equation*}
where $\lambda =\epsilon_{\mathcal{D}_{priv}}(h^*) + \epsilon_{\mathcal{D}_{pub}}(h^*) $ and $h^*$ is the ideal joint hypothesis minimizing the combined empirical risk: $h^* = argmin _{h\in H}\epsilon_{\mathcal{D}_{priv}}(h) + \epsilon_{\mathcal{D}_{pub}}(h)$.
\end{theorem}

\textit{Proof.} The proof of Theorem \ref{theorem2} builds on the lemma 3 in \cite{blitzer2007learning}, which demonstrates that for any hypotheses $h_1,h_2 \in \mathcal{H}$,
\begin{equation} \label{Eq11}
    |\epsilon_{\mathcal{D}_{priv}}(h_1,h_2)-\epsilon_{\mathcal{D}_{pub}}(h_1,h_2)| \leq \frac{1}{2}d_{\mathcal{H}\Delta\mathcal{H}}(\mathcal{D}_{priv}, \mathcal{D}_{pub}),
\end{equation}
and the triangle inequality for classification error \cite{crammer2008learning}, which demonstrates that for any hypothesis $h_1$,$h_2$,$h_3 \in \mathcal{H}$ with respect to $\mathcal{D}$, 
\begin{equation} \label{Eq12}
    \epsilon_{\mathcal{D}}(h_1,h_2) \leq \epsilon_{\mathcal{D}}(h_1,h_3) + \epsilon_{\mathcal{D}}(h_2,h_3).
\end{equation}
Then we can derive the bound of the difference between the empirical risk of the mixed dataset $\mathcal{D}_{mix}$ and the private dataset $\mathcal{D}_{priv}$ as
\begin{align}
    &|\epsilon_{\mathcal{D}_{mix}}(h,f_p) - \epsilon_{\mathcal{D}_{priv}}(h,f_p) |\notag \\
    =& (1-\alpha) |\epsilon_{\mathcal{D}_{priv}}(h,f_p) - \epsilon_{\mathcal{D}_{pub}}(h,f_p)|\notag \\
    =& (1\!-\!\alpha) \{|[\epsilon_{\mathcal{D}_{priv}}(h,f_p)\!-\!\epsilon_{\mathcal{D}_{priv}}(h,h^*)]
    \!+\![\epsilon_{\mathcal{D}_{pub}}(h,h^*)-\notag\\ &\epsilon_{\mathcal{D}_{pub}}(h,f_p)]
    + [\epsilon_{\mathcal{D}_{priv}}(h,h^*)-\epsilon_{\mathcal{D}_{pub}}(h,h^*)]|\}\notag \\
    \leq & (1\!-\!\alpha) [|\epsilon_{\mathcal{D}_{priv}}(h,f_p)-\epsilon_{\mathcal{D}_{priv}}(h,h^*)|
    \!+\!|\epsilon_{\mathcal{D}_{pub}}(h,f_p)-\notag\\ &\epsilon_{\mathcal{D}_{pub}}(h,h^*)|
    +|\epsilon_{\mathcal{D}_{priv}}(h,h^*)-\epsilon_{\mathcal{D}_{pub}}(h,h^*)|] \label{Eq13}\\
    \leq & (1-\alpha) [|\epsilon_{\mathcal{D}_{priv}}(h^*,f_p)|
    + |\epsilon_{\mathcal{D}_{pub}}(h^*,f_p)|+\notag \\
    &|\epsilon_{\mathcal{D}_{priv}}(h,h^*)-\epsilon_{\mathcal{D}_{pub}}(h,h^*)|] \label{Eq14} \\
    \leq& (1-\alpha)(\frac{1}{2}d_{\mathcal{H}\Delta\mathcal{H}}(\mathcal{D}_{priv}, \mathcal{D}_{pub}) +\lambda). \label{Eq15}
\end{align}
In the proof, Equation (\ref{Eq13}) is derived from the absolute value inequality, Equation (\ref{Eq14}) is derived from the triangle inequality introduced in Equation (\ref{Eq12}) and the last step is derived by substituting Equation (\ref{Eq11}) and $\lambda =\epsilon_{\mathcal{D}_{priv}}(h^*) + \epsilon_{\mathcal{D}_{pub}}(h^*)$.

The attacking model $\tilde{\theta}$ is trained on mixed dataset and test on victim's private dataset. According to Theorem \ref{theorem2}, we establish that the bound of the privacy risk, quantified by the performance of $\tilde{\theta}$, is bounded by $(1-\alpha)(\frac{1}{2}d_{\mathcal{H}\Delta\mathcal{H}}(\mathcal{D}_{priv}, \mathcal{D}_{pub}) +\lambda)$. When fixing $\mathcal{D}_{priv}$ and $\mathcal{D}_{pub}$, this bound is exclusively related to weighting parameter $\alpha$. Notably, the bound drops to 0 as $\alpha$ increases to 1. This suggests that, when $\alpha$ increases, i.e. the mixed public data gets closer to the private data, the attacking model performs better on the private test dataset and more private information of the model, which is trained on its logits and the public dataset, is leaked.

Here we briefly delve into the underlying cause of this privacy risk in logit-based FL. As depicted in Figure \ref{fig2}, we observe that local model training derives benefit from the knowledge contained in the ensemble logit, obtained through the aggregation of local predicted logits on the public data. While a more informative local logit contributes to a more informative ensemble logit, it concurrently exposes more privacy to the adversary, as illustrated in Theorem \ref{theorem2}. Our observation aligns with empirical results, as demonstrated in Figure \ref{figure3} in Section \ref{exp_distance}.



\section{Defending against the Privacy Risk}
Our observation in Section \ref{section-3} shows that the privacy risk in logit-based FL mainly comes from the logit. In this section, we propose a defense strategy, named Federated Logit Perturbation (FedLP), that perturbs the transmitted logits of local models to defend against this privacy risk.

\subsection{Design of FedLP}
\textbf{Defense Objective} The defender (i.e. the clients) has two objectives. First, the defender aims to prevent an adversary from being able to replicate the functionality of its private model:
\begin{equation}\label{Eq3}
    {\rm min}_{\tilde{\theta}} \mathbb{E}_{x \sim \mathcal{D}_{priv}} Acc\left(f(x, \tilde{\theta})\right),
\end{equation} 
where $f(x,\tilde{\theta})$ denotes the functionality of the attacking model.

Second, the defender aims to preserve the training performance of the logit-based FL protocol so that the perturbation scale should be bounded by a non-negative constant $\gamma$:
\begin{equation}\label{Eq4}
   ||p-p'||_b \leq \gamma,
\end{equation} 
\noindent where $p$ is the original logit on the public data, $p'$ is the corresponding perturbed logit, $\gamma >0$ is a pre-determined constant parameter and $||.||_b$ denotes the $L_b$ norm. 

We note that the defender has no access to the adversary's model and may be even unaware that it is under attack since the attack happens at the server side. Therefore, the defender has to prevent the privacy risk from the semi-honest server during the whole training process.

\textbf{Defense Problem} Combining Equation (\ref{Eq3}) and (\ref{Eq4}), we can formulate a defense problem for the defender. However, this problem cannot be directly solved since the attacking model parameters and its training details are unknown to the defender. Therefore, we need to approximate the first objective from the perspective of the defender. 

The first step is to estimate the attacking model $\tilde{\theta}$. As the goal of $\tilde{\theta}$ is to approximate the defender's model $\theta$ in each training iteration, we estimate the attacking model obtained from the last iteration to be the same as the defender's model $\theta$ in the current iteration $T$:
\begin{equation*}
    \tilde{\theta}_{T-1}' = \theta_{T-1},
\end{equation*}

\noindent where  $\tilde{\theta}_{T-1}'$ and $\theta_{T-1}$ are the estimated attacking model and the defender's model in the last training iteration respectively.

Without loss of generality, we assume that the adversary optimizes its attacking model through the gradient of an empirical loss on the public data, which is the most widely used optimization method in deep learning \cite{oliinyk2020first}. The gradient of an empirical loss with respect to parameter $\theta$ can be expressed as
\begin{equation} \label{Eq6}
   G(\theta,p) = \nabla_{\theta} \mathcal{L}_{CE}(f(x,\theta),p).
\end{equation}

Based on the above assumptions, we restate the objectives of the defender as maximally changing the updated gradients of the estimated attacking model with minimum perturbation on the logit. That is, we can rewrite the first objective of the defender in iteration $T$ as maximizing the distance between the gradients of the estimated attacking model updated through the original logit and the perturbed logit, in terms of $L_a$ norm:
\begin{align} \label{Eq7}
    &{\rm max}_{p'} ||G(\tilde{\theta}'_{T-1},p') - G(\tilde{\theta}'_{T-1},p) ||_a \notag\\
    =&{\rm max}_{p'} ||G(\theta_{T-1},p') - G(\theta_{T-1},p) ||_a
\end{align}
where $G(\theta,p)$ is the gradient of the empirical loss with respect to the parameters $\theta$.

Combining Equation (\ref{Eq7}) and (\ref{Eq4}), we therefore reformulate the defense problem as a constrained optimization problem:
\begin{align}
{\rm max}_{p'} ||&G(\theta_{T-1},p') - G(\theta_{T-1},p) ||_a \label{Eq8}\\
&{\rm  s.t.} \  ||p-p'||_b \leq \gamma,\label{Eq9}
\end{align}
which allows the defender to trade off the utility and privacy in logit-based FL training. Note that we can form multiple defense problems and corresponding defense strategies with different selections of $(a,b)$. In this paper, we set $(a,b)=(2,1)$ as default and compare with other settings empirically in Section \ref{section-5.3}.

\textbf{Solution of FedLP} \label{section4-3} Deep learning models usually involves millions of parameters and thus solving Equation (\ref{Eq8}) s.t. Equation (\ref{Eq9}) with respect to each sample in the public dataset requires a large computational cost for clients, which is unaffordable for local devices in practice. Here, we give a simple heuristic solver to circumvent this computational issue. We perturb the logit $p$ on the public dataset in each dimension of itself by $Z$:
\begin{equation*}
    p' = p + Z \cdot e_j,
\end{equation*}
where $e_j$ denotes a one-hot vector with 1 in the $j$-th dimension of $p$ and 0's elsewhere. Then we select the one giving the largest perturbation in Equation (\ref{Eq8}). 
The local training with our defense in one iteration is given in Algorithm~\ref{alg_FedLP}. Note that Algorithm~\ref{alg_FedLP} is executed entirely on the client's local machine except for the communication step, where only the perturbed logits $p'$ are transmitted to the server. By perturbing the shared logits along the direction that maximizes the estimated gradient discrepancy, FedLP aims to interfere with the adversary's model stealing process while limiting the impact on normal training. We theoretically analyze the gradient disruption and utility preservation properties of FedLP in Section~\ref{analysis_fedlp}, and empirically evaluate its defense performance in Section~\ref{section-5.3}.

We provide a brief analysis of the computational complexity associated with FedLP. Since the number of logit dimensions $C$ is fixed for a given classification task and is typically much smaller than $|\mathcal{D}_{pub}|$, the computational complexity of FedLP, as outlined in Algorithm~\ref{alg_FedLP}, is $O(|\mathcal{D}_{pub}|)$. Notably, our defense operates solely on the local client, meaning there is no additional burden on communication costs. Furthermore, varying the noise scale $\sigma$ does not result in higher computational or communication costs, as shown in Algorithm~\ref{alg_FedLP}.

\textbf{Noise Selection.}
The perturbation magnitude $Z$ can be sampled from different distributions, such as Gaussian or Laplace distributions. In this paper, we adopt Gaussian noise $Z\sim\mathcal{N}(0,\sigma^2)$, where $\sigma$ denotes the noise standard deviation. 
Since Gaussian noise is unbounded, the constraint in Equation~(\ref{Eq9}) represents the intended perturbation budget, while the actual perturbation magnitude is controlled in expectation through $\sigma$. FedLP applies the sampled noise to the selected logit dimension, rather than perturbing all dimensions independently. Increasing $\sigma$ increases the expected perturbation magnitude, allowing FedLP to adjust the strength of gradient disruption while potentially affecting model utility. Theoretical analyses of the gradient sensitivity, directional perturbation advantage, and distillation target distortion are provided in Section~\ref{analysis_fedlp}.

\begin{algorithm}[h]
\caption{Local Training with FedLP in one iteration}
\begin{algorithmic}[1]
\label{alg_FedLP}
\REQUIRE Public dataset $\mathcal{D}_{pub}$, local dataset $\mathcal{D}_{priv}:\{x,y\}$, local model in last iteration $\theta_{T-1}$, noise standard deviation $\sigma$.
\ENSURE Updated local model $\theta_T$.

\STATE \textbf{Local Training:}
\STATE Train local model with $\mathcal{D}_{priv}$ and update $\theta$

\STATE \textbf{Logit Ensemble:}
\FOR{$x_i$ in $\mathcal{D}_{pub}$}
    \STATE $p_i \leftarrow \theta_T(x_i)$
\ENDFOR

\STATE $Z \sim \mathcal{N}(0,\sigma^2)$
\FOR{each dimension $e_j$ of $p_i$}
    \STATE $p'_{i,j} \leftarrow p_i + Z \cdot e_j$, for all $x_i \in \mathcal{D}_{pub}$
    \STATE Calculate $D_j(Z) \leftarrow \sum\limits_{x_i \in \mathcal{D}_{pub}}
    \|G_i(\theta_{T-1},p'_{i,j}) - G_i(\theta_{T-1},p_i)\|_2$
    according to Equation~(\ref{Eq6})
\ENDFOR

\STATE $j^* \leftarrow \arg\max_j D_j(Z)$
\STATE $p'_i \leftarrow p_i + Z \cdot e_{j^*}$, for all $x_i \in \mathcal{D}_{pub}$
\STATE Upload $p'=\{p'_i\}_{x_i\in\mathcal{D}_{pub}}$ to the server and then obtain the corresponding ensemble logits from the server

\STATE \textbf{Distillation:}
\STATE Train local model with ensemble logits on $\mathcal{D}_{pub}$ and update $\theta$
\STATE \textbf{return} $\theta_T$

\end{algorithmic}
\end{algorithm}

\subsection{Theoretical Analysis of FedLP}
\label{analysis_fedlp}

As formulated in Equations~(\ref{Eq8})--(\ref{Eq9}), FedLP aims to enhance privacy protection by maximizing the gradient discrepancy of the estimated attacking model while limiting logit perturbation to preserve model utility. Intuitively, a larger deviation from the gradient induced by the original logits may hinder the attacking model from accurately approximating the victim model, thereby reducing the risk of privacy leakage. Although gradient discrepancy does not directly determine the final attack performance, it provides a useful surrogate for analyzing the effectiveness of FedLP. We therefore theoretically analyze the privacy and utility properties of FedLP by characterizing the gradient sensitivity of logit perturbation, establishing the advantage of directional perturbation, and bounding the distortion of knowledge distillation targets.

\textbf{Gradient Sensitivity Analysis.}
We first analyze the effect of FedLP on the gradients of the estimated attacking model. Following Equations~(\ref{Eq8})--(\ref{Eq9}), we consider the default setting $(a,b)=(2,1)$, which maximizes the gradient discrepancy under an $L_1$ logit perturbation constraint.

For each public sample $x_i\in\mathcal{D}_{pub}$, let $p_i$ denote its original logit and $G_i(\theta_{T-1},p_i)$ denote the corresponding gradient defined in Equation~(\ref{Eq6}). Let $s(p_i)=\operatorname{softmax}(p_i)$, and define
\begin{equation}
A_i=J_{\theta_{T-1}}(x_i)^\top
\left(\operatorname{diag}(s(p_i))
-s(p_i)s(p_i)^\top\right),
\end{equation}
where $J_{\theta_{T-1}}(x_i)$ denotes the Jacobian of the model logits with respect to its parameters.

\begin{theorem}[Gradient Sensitivity Characterization]
\label{thm:gradient_sensitivity}
For the perturbation $p_i'=p_i+Ze_j$, where $Z\sim\mathcal{N}(0,\sigma^2)$ is shared across the public dataset, the cumulative gradient discrepancy satisfies
\begin{equation}
\begin{aligned}
D_j(Z)
&=\sum_{x_i\in\mathcal{D}_{pub}}
\left\|
G_i(\theta_{T-1},p_i+Ze_j)
-G_i(\theta_{T-1},p_i)
\right\|_2\\
&=|Z|\sum_{x_i\in\mathcal{D}_{pub}}
\|A_ie_j\|_2+O(Z^2),
\end{aligned}
\end{equation}
as $Z\rightarrow 0$.
\end{theorem}

Theorem~\ref{thm:gradient_sensitivity} shows that gradient disruption depends on the perturbation magnitude, model Jacobian, prediction distribution, and selected logit direction. This motivates FedLP to select the direction yielding the largest cumulative gradient discrepancy.

\textbf{Directional Perturbation Advantage.}
We next compare FedLP with full-dimensional Gaussian perturbation under the same expected $L_1$ perturbation budget. Let $C$ denote the number of logit dimensions. FedLP selects $j^*=\arg\max_j D_j(Z)$ and applies $Ze_{j^*}$. For comparison, consider a Gaussian perturbation
\begin{equation}
\xi\sim\mathcal{N}
\left(0,\frac{\sigma^2}{C^2}I_C\right),
\end{equation}
where the same vector $\xi$ is applied across the public dataset. Both mechanisms have the same expected $L_1$ perturbation magnitude:
\begin{equation}
\mathbb{E}\|Ze_{j^*}\|_1
=
\mathbb{E}\|\xi\|_1
=
\sigma\sqrt{\frac{2}{\pi}}.
\end{equation}

\begin{theorem}[Directional Perturbation Advantage]
\label{thm:directional_advantage}
Let $D_{\mathrm{FedLP}}=D_{j^*}(Z)$ and define the cumulative gradient discrepancy under Gaussian perturbation as
\begin{equation}
D_{\mathrm{G}}
=
\sum_{x_i\in\mathcal{D}_{pub}}
\left\|
G_i(\theta_{T-1},p_i+\xi)
-G_i(\theta_{T-1},p_i)
\right\|_2.
\end{equation}
Under the same expected $L_1$ perturbation budget,
\begin{equation}
\begin{aligned}
&\liminf_{\sigma\to 0}
\frac{
\mathbb{E}[D_{\mathrm{FedLP}}]
-\mathbb{E}[D_{\mathrm{G}}]
}{\sigma}\\
&\quad\geq
\sqrt{\frac{2}{\pi}}
\left(
\max_j\sum_{x_i\in\mathcal{D}_{pub}}\|A_ie_j\|_2
-\frac{1}{C}\sum_{j=1}^{C}
\sum_{x_i\in\mathcal{D}_{pub}}\|A_ie_j\|_2
\right)
\geq 0.
\end{aligned}
\end{equation}
\end{theorem}

Theorem~\ref{thm:directional_advantage} establishes that FedLP achieves no smaller first-order expected gradient disruption than full-dimensional Gaussian perturbation under an equal expected $L_1$ budget. By concentrating the perturbation on the most influential logit direction, FedLP can more effectively disrupt the estimated attacker's gradient updates while controlling the perturbation magnitude.

The proofs of Theorems~\ref{thm:gradient_sensitivity} and~\ref{thm:directional_advantage} are provided in the Appendix.

\textbf{Utility Preservation Analysis.}
While FedLP aims to increase gradient discrepancy for privacy protection, the perturbation of shared logits may also affect the knowledge distillation process. We therefore analyze the distortion of the distillation targets induced by FedLP.

\begin{theorem}[Bounded Distillation Target Distortion]
\label{thm:utility_distortion}
Let $p$ and $p'$ denote the original and perturbed logits, respectively, and let $s(p)=\operatorname{softmax}(p)$. For any logit perturbation, the KL divergence between the corresponding prediction distributions satisfies
\begin{equation}
\begin{aligned}
D_{\mathrm{KL}}(s(p)\|s(p'))
&\leq \frac{1}{4}\|p-p'\|_2^2\\
&\leq \frac{1}{4}\|p-p'\|_1^2.
\end{aligned}
\end{equation}
Consequently, the $L_1$ perturbation constraint in Equation~(\ref{Eq9}) guarantees
\begin{equation}
D_{\mathrm{KL}}(s(p)\|s(p'))
\leq \frac{\gamma^2}{4}.
\end{equation}
For FedLP with $p'=p+Ze_{j^*}$ and $Z\sim\mathcal{N}(0,\sigma^2)$, the expected distortion satisfies
\begin{equation}
\mathbb{E}_Z
\left[D_{\mathrm{KL}}(s(p)\|s(p'))\right]
\leq \frac{\sigma^2}{4}.
\end{equation}
\end{theorem}

Theorem~\ref{thm:utility_distortion} shows that the logit perturbation budget provides an upper bound on the distortion of knowledge distillation targets. In particular, the expected distortion introduced by FedLP is controlled by the noise variance, regardless of the selected perturbation direction. Although this bound does not directly guarantee the final model accuracy, it provides a theoretical explanation for controlling utility distortion while disrupting the estimated attacker's gradient updates.

The proof is provided in the Appendix.

\section{Experiment}
In this section, we empirically validate the effectiveness of AdaMSA against logit-based FL and FedLP against AdaMSA. Also, we assess various factors influencing the attack and defense performance. Specifically, we aim to answer two main questions through our experiments: 
\begin{itemize}
    \item \textbf{RQ1 Attack Evaluation:} Is the proposed AdaMSA effective against logit-based FL? 
    \item \textbf{RQ2 Defense Evaluation:} Is the proposed FedLP effective against AdaMSA, i.e. can it achieve a better utility and privacy trade-off?
\end{itemize}

\subsection{Experimental Settings}\label{settings} 
To evaluate our proposed attack and defense, we construct three experimental settings for the image classification task: 
\begin{itemize}
\item[$\bullet$] \textbf{Close-world:} Following the same experimental setting in Fedmd \cite{li2019fedmd}, we first construct a simple Close-world setting that public data and private data are similar. We use MNIST\footnote{https://yann.lecun.com/exdb/mnist/} \cite{lecun1998gradient} without labels as the unlabeled public dataset and EMNIST\footnote{https://www.nist.gov/itl/products-and-services/emnist-dataset} \cite{cohen2017emnist} as the private dataset. We randomly select 300 images (30 for each class) in EMNIST as the private training dataset and then identically distributed to 10 clients. The rest of the private dataset (EMNIST) is used as the private test dataset.

\item[$\bullet$] \textbf{Open-world-CF:} To construct a cross-domain experimental setting, we first discard the labels in CIFAR10\footnote{https://www.cs.toronto.edu/~kriz/cifar.html} \cite{krizhevsky2009learning} and utilize it as the unlabeled public dataset. We use SVHN\footnote{http://ufldl.stanford.edu/housenumbers/} \cite{netzer2011reading} as the private dataset, which comprises house numbers from Google Street View images, and is unrelated to the images in CIFAR-10. We randomly select 1000 images in each class of SVHN as the private training dataset and identically distributed to 10 clients. The rest of the private dataset (SVHN) is used as the private test dataset. 

\item[$\bullet$] \textbf{Open-world-TI:} Similar to Open-world-CF, we select the first 20 classes in TinyImagenet\footnote{https://www.kaggle.com/c/tiny-imagenet/data} \cite{le2015tiny} with 500 images per class without labels as the unlabeled public dataset and SVHN \cite{netzer2011reading} as the private dataset. Then we randomly select 1000 images in each class of SVHN as the private training dataset and identically distributed to 10 clients. The rest of the private dataset (SVHN) is used as the private test dataset. 
\end{itemize}

\textbf{Implementation} We adopt CNN \cite{lecun1989handwritten} as backbones for all clients' models. We use a 2-layer CNN with (128,256) parameters as the attacking model and a CNN with (128,256) parameters initialized from a different random seed as the victim model to report the main result in Table \ref{table2}. All models are optimized by Adam with a 0.01 learning rate. For a fair comparison, we set the perturbation scale $\gamma$ to be the same as 0.01 for DP-G, DP-L and our defense in Table \ref{table2} and $T_0, w_0$ is set to be 3 and 0.5 respectively. For Close-world setting, the batch size of local training on private data is 10 and on public data is 128. The number of communication rounds is 10. For Open-world settings, the batch size of local training on the private dataset is 32 and on the public dataset is 256. The number of communication rounds is 30. We repeat each experiment three times and report the average accuracy as the results.

\textbf{Baselines} We compare our proposed attack against \textbf{MSA} \cite{tramer2016stealing} and \textbf{Naïve AdaMSA} ($T_0=0$) baselines in three experimental settings. For MSA baseline \cite{tramer2016stealing}, we let the attacking model learn from the victim's logit in the final iteration. For Naïve AdaMSA baseline, we train an attacking model by approximating the victim's current logit in each round and choose the attacking model with highest accuracy.

For defense evaluation, we compare our proposed FedLP to five state-of-the-art baseline defenses:
1) \textbf{Unprotected} \cite{li2019fedmd}: We follow the training process of the general logit-based FL approach \cite{li2019fedmd}. For Unprotected defense, the public and private data are drawn from the same dataset, i.e., we set $\alpha = 1$ for the mixed dataset. Specifically, we randomly select a subset of the private test dataset without labels to serve as the public dataset, with the remaining portion being used for testing. No defense mechanisms are applied in this setting; 2) \textbf{Cross-domain} \cite{lin2020ensemble}: In Cross-domain defense, the public dataset is selected from another irrelevant and insensitive domain to prevent privacy leakage (i.e. we set $\alpha = 0$ for the mixed dataset); 3) \textbf{One-shot} \cite{li2020practical}:  In One-shot defense, we utilize one-shot distillation on unlabeled and domain-robust public data. Specifically, clients only communicate once with server. 4) Differential Privacy (DP): Recent works \cite{gong2022preserving,sattler2021fedaux} adopt Gaussian and Laplacian DP and add noise to the transmitted updates. We conduct their strategies on the unprotected baseline as \textbf{DP-G} \cite{sattler2021fedaux} and \textbf{DP-L} defense \cite{gong2022preserving}, respectively.

\textbf{Evaluation Metric} We gauge the attack performance by assessing the prediction accuracy (Acc) of the attacking model on the victim's private test dataset. A greater Acc value indicates a more formidable attacking model.

For defense evaluation, we evaluate all defenses on a utility loss vs. privacy loss curve at various points of the defenses. The utility loss $\Delta U$ of the defense is defined as
\begin{equation}
    { \Delta U} = {U_0} - { U_d},
\end{equation}

\noindent where $U_0$ is the accuracy of local model in unprotected baseline and $U_d$ is the accuracy of local model under different defense. The privacy loss is defined as the prediction accuracy of the attacking model on the defender's private test dataset.

\begin{centering}
\setlength{\tabcolsep}{3mm}{
\begin{table*}[t]
\caption{Attack performance of AdaMSA and two attack baselines on the victim model with various defense baselines in Close-world, Open-world-CF and Open-world-TI settings. Victim denotes the performance of the victim model.}
\label{table2}
\centering
\resizebox{\linewidth}{!}{
\begin{tabular}{cccccc}
\toprule
Setting &  Defense &  Victim (\%) & MSA\cite{tramer2016stealing} (\%) & Naïve AdaMSA (\%) & AdaMSA($T_0$=3) (\%)\\
\hline
\multirow{4}{*}{Close-world}  &  {Unprotected \cite{li2019fedmd}}  & 83.43 $\pm$ 1.07 & 80.10 $\pm$ 1.31 & 80.09 $\pm$ 1.35 & \textbf{83.79} $\pm$ 1.20 \\
 & {Cross-domain \cite{lin2020ensemble}} & 82.97 $\pm$ 0.89 & 79.43 $\pm$ 1.07 & 79.12 $\pm$ 1.30 &\textbf{82.68} $\pm$ 1.01 \\
 & {One-shot \cite{gong2021ensemble}} & 68.23 $\pm$ 1.21 & 65.03 $\pm$ 0.89 & 66.15 $\pm$ 0.73 &\textbf{67.18} $\pm$ 1.17 \\
 & {DP-G \cite{sattler2021fedaux}} & 74.68 $\pm$ 2.15 & 71.69 $\pm$ 1.14 & 72.07 $\pm$ 2.32 & \textbf{74.79} $\pm$ 1.20  \\
 & {DP-L \cite{gong2022preserving}} & 75.65 $\pm$ 1.97 & 72.01 $\pm$ 1.21 & 71.53 $\pm$ 1.44 & \textbf{75.53} $\pm$ 1.09  \\
\hline
\multirow{4}{*}{Open-world-CF}  &  {Unprotected \cite{li2019fedmd}} & 71.11 $\pm$ 0.98 & 68.77 $\pm$ 1.13 & 68.99 $\pm$ 1.02 & \textbf{71.25} $\pm$ 0.71 \\
 & {Cross-domain \cite{lin2020ensemble}}  & 53.57 $\pm$ 1.13 & 52.21 $\pm$ 1.06 & 52.34 $\pm$ 1.47 &\textbf{53.78} $\pm$ 0.79  \\
 & {One-shot \cite{gong2021ensemble}}  & 63.41 $\pm$ 1.09 & 60.03 $\pm$ 0.98 & 59.66 $\pm$ 0.99 &\textbf{62.49} $\pm$ 0.91  \\
 & {DP-G \cite{sattler2021fedaux}}  & 65.89 $\pm$ 1.77 & 62.20 $\pm$ 1.39 & 62.33 $\pm$ 1.56 &\textbf{65.61} $\pm$ 1.31  \\
  & {DP-L \cite{gong2022preserving}} & 65.22 $\pm$ 1.96 & 65.97 $\pm$ 1.28 & 66.07 $\pm$ 1.25 &\textbf{67.64} $\pm$ 1.22 \\
\hline
\multirow{4}{*}{Open-world-TI}  &  {Unprotected \cite{li2019fedmd}} & 71.11 $\pm$ 0.98 & 69.34 $\pm$ 0.97 & 69.68 $\pm$ 1.01 & \textbf{72.13} $\pm$ 1.02 \\
 & {Cross-domain \cite{lin2020ensemble}}& 51.14 $\pm$ 1.10 & 51.17 $\pm$ 1.05 & 50.43 $\pm$ 1.17 & \textbf{53.54} $\pm$ 0.80 \\
 & {One-shot \cite{gong2021ensemble}} & 63.09 $\pm$ 1.14 & 62.36 $\pm$ 1.03 & 62.09 $\pm$ 0.86 & \textbf{64.41} $\pm$ 0.99 \\
 & {DP-G \cite{sattler2021fedaux}} & 65.76 $\pm$ 2.01 & 62.77 $\pm$ 1.20 & 63.01 $\pm$ 1.33 & \textbf{65.02} $\pm$ 1.18 \\
  & {DP-L \cite{gong2022preserving}} & 65.09 $\pm$ 1.98 & 66.11 $\pm$ 1.13 & 66.32 $\pm$ 1.51 & \textbf{68.41} $\pm$ 1.04 \\
\bottomrule
\end{tabular}}
\end{table*}
}
\end{centering}

\subsection{Attack Performance Evaluation} \label{attack_eval}
\textbf{Main Results} Table \ref{table2} shows the performance of MSA, Naïve AdaMSA and AdaMSA in three settings. We observe that AdaMSA demonstrates significant improvements of up to 4.00, 3.41, and 3.11 percentage points in the respective settings, compared to Naïve AdaMSA and MSA baselines. This highlights the effectiveness of our attack design, utilizing historical logits to generate a more informative target logit, and thereby enhance the attack performance. Moreover, it is observed that AdaMSA consistently matches the performance of the victim model across various defense baselines in three settings. This reaffirms its capability to successfully steal the functionality of the victim's private model.

It is worth mentioning that we observed the performance of the attacking model surpassing that of the victim model in certain baselines. The surprising phenomenon can be attributed to the fact that the training process of the attacking model can be viewed as a form of self-distillation, which has been demonstrated to improve the model's generalization ability on test data \cite{allen-zhu2023towards}.

\textbf{Under Limited Resource Assumption} We would like to clarify that we follow the general assumption in FL that the server is powerful and has virtually unlimited computational resources. Under these conditions, AdaMSA($T_0=3$) is the preferred choice. However, if the server's computational resources are limited, the Naïve AdaMSA may offer a better balance between performance and efficiency, as AdaMSA($T_0=3$) requires 
three times the computational resources.

\begin{figure}[t] \label{figure3}
\begin{minipage}{0.4\linewidth} 
  \includegraphics[width=\linewidth]{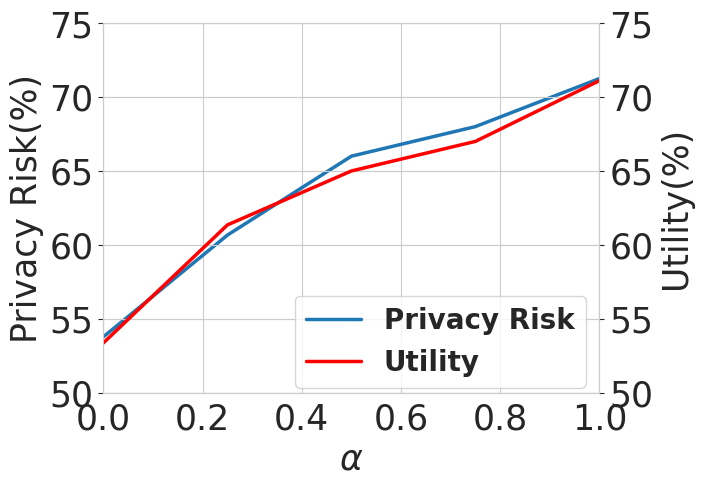}
\end{minipage}
\begin{minipage}{0.4\linewidth}
  \includegraphics[width=\linewidth]{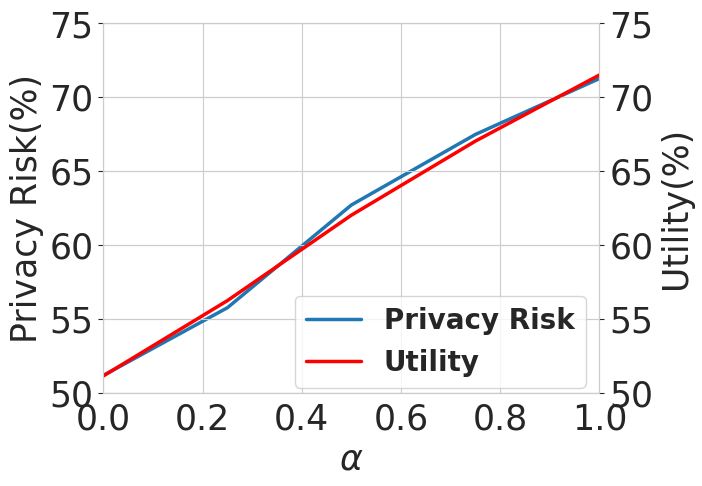}
\end{minipage}
\caption{The relation between the utility/privacy risk and $\alpha$ in Open-world-CF (left) and Open-world-TI (right) settings. Privacy risk is measured by the attack performance of AdaMSA.}
\label{figure3}
\end{figure}

\begin{table}[t] \centering
\captionof{table}{The effect of combining historical logits on the attack performance of AdaMSA.}
\label{table3}
\resizebox{0.35\linewidth}{!}{
\begin{tabular}{ccc}
\toprule
$T_0$ & \multicolumn{2}{c}{AdaMSA (\%)} \\
\cline{2-3}
& Open-world-CF & Open-world-TI \\
\hline
0 & 68.99 $\pm$ 1.02 & 69.68 $\pm$ 1.01\\
1 & 70.11 $\pm$ 0.94 & 70.51 $\pm$ 1.16 \\
2 &  70.40 $\pm$ 1.05 & 71.88 $\pm$ 1.15 \\
3 &  \textbf{71.25} $\pm$ 0.71 & \textbf{72.13} $\pm$ 1.02\\
4 &  71.09 $\pm$ 1.46 & 71.88 $\pm$ 1.35 \\
\bottomrule
\end{tabular}}
\end{table}

\subsection{Ablation Studies}
In this section, we conduct ablation analysis of the attack performance on the impact of the distance between private and public datasets, combining historical logits, data Heterogeneity and model architectures. 

\textbf{Effect of the Distance Between Private and Public Datasets} \label{exp_distance} To evaluate the effect of this factor, we construct several mixed datasets (see section \ref{section3-3} for detailed illustration) in two Open-world settings as the public datasets through varying the weighting parameter $\alpha$. The results are reported in Figure \ref{figure3}. It can be observed that increasing the value of $\alpha$ (i.e. decrease the distance), will increase the utility at the cost of increasing privacy risk. Furthermore, \textit{this inherent privacy risk persists even when public data is unrelated to private data}. This result indicates that the utility and privacy in logit-based FL are indeed two sides of a coin. That is, a more informative local logit results in a more informative ensemble logit to supervise the local model training, meanwhile it also exposes more privacy to the adversary. This result aligns with our analysis in Section \ref{section3-3}.

\textbf{Effect of Combining Historical Logits} To evaluate the effect of combining historical logits, we vary $T_0$ and test on the same victim model in two open-world settings. From the results in Table \ref{table3}, we observe that, when $T_0$ increases from 0 to 3, the attack performance gradually increases by 2.26\% in Open-world-CF setting and 2.45\% in Open-world-TI setting, demonstrating that combining more historical logits indeed improves the attack performance. The rationale behind this improvement is that the historical predictions close to the current round can be viewed as the different views of the victim model on the public data. Therefore, our designed target logit can be benefited from the ensemble of these muti-view predictions, improving model generalization ability on the test data \cite{allen2020towards}. However, when $T_0$ increases to 4, the performance drops because the predictions in the early rounds may not be sufficiently trained and, instead of boosting, can act as a detriment. 

\begin{table*}[t]
    \centering
    \caption{The effect of data heterogeneity on the attack performance of AdaMSA in Open-world-CF setting.}
    \label{table4}
    \resizebox{0.95\linewidth}{!}{
    \begin{tabular}{ccccccccc}
        \toprule
       \multirow{2}{*}{Defense}  &  \multicolumn{2}{c}{$\beta=0.5$} & \multicolumn{2}{c}{$\beta=1$} & \multicolumn{2}{c}{$\beta=10$} & \multicolumn{2}{c}{$\beta=100$} \\
       \cline{2-3} \cline{4-5}  \cline{6-7} \cline{8-9}
       & Victim (\%)& AdaMSA (\%) & Victim (\%) & AdaMSA (\%) & Victim (\%) & AdaMSA (\%) & Victim (\%) & AdaMSA (\%)\\
       \hline
        Unprotected \cite{li2019fedmd} & 63.99 $\pm$ 0.85 & 64.02 $\pm$ 1.16 & 66.32 $\pm$ 0.79 & 66.26 $\pm$ 1.05 & 70.36 $\pm$ 1.25 & 69.79 $\pm$ 1.04 & 71.09 $\pm$ 1.21 & 71.21 $\pm$ 0.99\\
        Cross-domain \cite{lin2020ensemble}& 47.30 $\pm$ 0.98 & 47.57 $\pm$ 1.12 & 49.04 $\pm$1.83 & 49.11 $\pm$ 1.64 & 54.01 $\pm$ 0.99 & 54.03 $\pm$ 1.33 & 54.24 $\pm$ 1.92 & 54.20 $\pm$ 1.65\\
        One-shot \cite{gong2021ensemble}& 55.93 $\pm$ 1.34 & 55.97 $\pm$ 1.04 & 57.30 $\pm$ 1.46 & 57.42 $\pm$ 1.02 & 62.56 $\pm$ 1.23 & 62.66 $\pm$ 1.02 & 63.06 $\pm$ 0.55 & 63.25 $\pm$ 0.97\\
        DP-G \cite{sattler2021fedaux}& 57.21 $\pm$ 0.49 & 56.88 $\pm$ 0.95 & 59.48 $\pm$ 1.55 & 60.01 $\pm$ 0.87 & 63.75 $\pm$ 1.50 & 64.11 $\pm$ 1.39 & 65.88 $\pm$ 1.43 & 66.89 $\pm$ 1.03 \\
        DP-L \cite{gong2022preserving}& 57.63 $\pm$ 1.40 & 57.54 $\pm$ 1.51 & 59.38 $\pm$ 2.12 & 59.90 $\pm$ 1.52 & 63.60 $\pm$ 1.82 & 63.59 $\pm$ 1.71 & 65.59 $\pm$ 1.96 & 65.99 $\pm$ 1.97 \\
        \bottomrule
    \end{tabular}}
\end{table*}


\textbf{Effect of Data Heterogeneity} To evaluate the effect of data heterogeneity, we split the training data according to a Dirichlet distribution following \cite{hsu2019measuring} in Openworld-CF setting. The non-iid level of data is controlled by the Dirichlet parameter $\beta$. From the result reported in Table \ref{table4}, we observe that: 1) increasing the non-iid level (i.e. decrease $\beta$) will decrease the utility of local model as well as the attacking model performance; 2) AdaMSA achieves similar performance compared to the victim model, indicating AdaMSA is still effective in the non-iid setting.

\textbf{Effect of Model Heterogeneity} We vary the victim model parameters and architectures in the Open-world-CF setting to evaluate the effect of victim model heterogeneity on the attack performance. The results are reported in Table \ref{table6}. We also find a similar result in the Open-world-TI setting, which is omitted due to the space limit. For a fair comparison, we perform the experiments in the unprotected baseline. The victim models vary from CNN to Alexnet \cite{krizhevsky2017imagenet} with different hyperparameters. The details are reported in Table \ref{table6}. The results show that AdaMSA can achieve similar accuracy of the victim models under different model hyperparameters. This indicates that AdaMSA is robust to heterogeneous victim models in our setting. 

\begin{centering}
\begin{table}[t] \centering
\captionof{table}{The effect of model heterogeneity on the attack performance of AdaMSA in Open-world-CF setting.}\label{table6}
\resizebox{0.6\linewidth}{!}{
\begin{tabular}{ccccc}
\toprule
Model & Layer & Parameter & Victim Acc(\%) & Attack Acc(\%) \\
\hline
CNN & 2 & 128,256 & 71.11 $\pm$ 1.21 & 69.33 $\pm$ 0.99 \\
CNN & 2 & 128,512 & 70.15 $\pm$ 1.14 & 70.28 $\pm$ 1.02\\
CNN & 2 & 256,512 & 70.56 $\pm$ 1.30 & 70.42 $\pm$ 0.79 \\
CNN & 3 & 64,128,256 & 71.11 $\pm$ 0.98 & 71.25 $\pm$ 0.71 \\
CNN & 3 & 128,192,256 & 69.64 $\pm$ 1.28 & 69.19 $\pm$ 1.13 \\
Alexnet \cite{krizhevsky2017imagenet} & 8 & - & 75.10 $\pm$ 1.29 & 75.33 $\pm$ 1.24\\
\bottomrule
\end{tabular}}
\end{table}
\end{centering}

\subsection{Defense Performance Evaluation} \label{section-5.3}
\begin{figure*}[h]\centering 
\subfigure[Close-world]{
\includegraphics[width=0.3\linewidth]{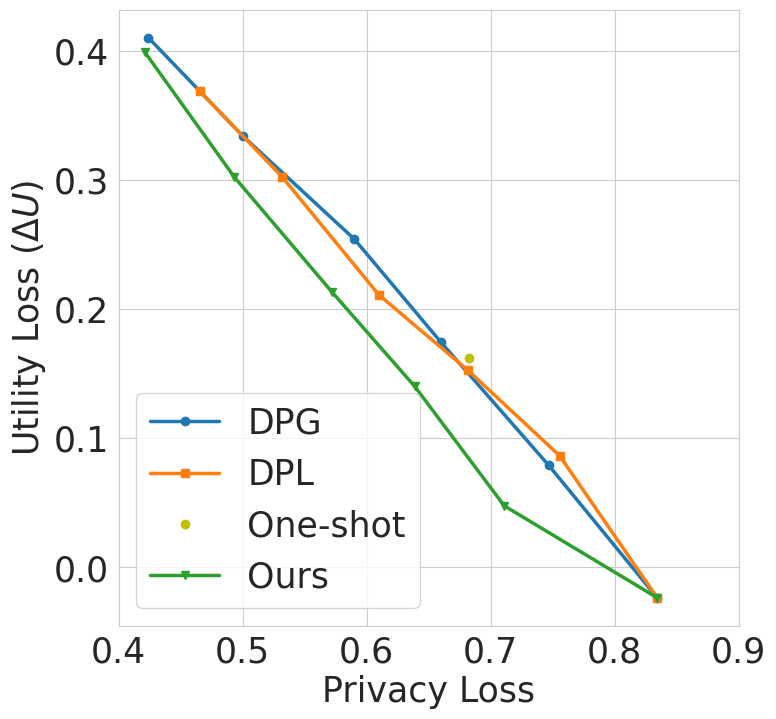}
}
\subfigure[Open-world-CF]{
\includegraphics[width=0.31\linewidth]{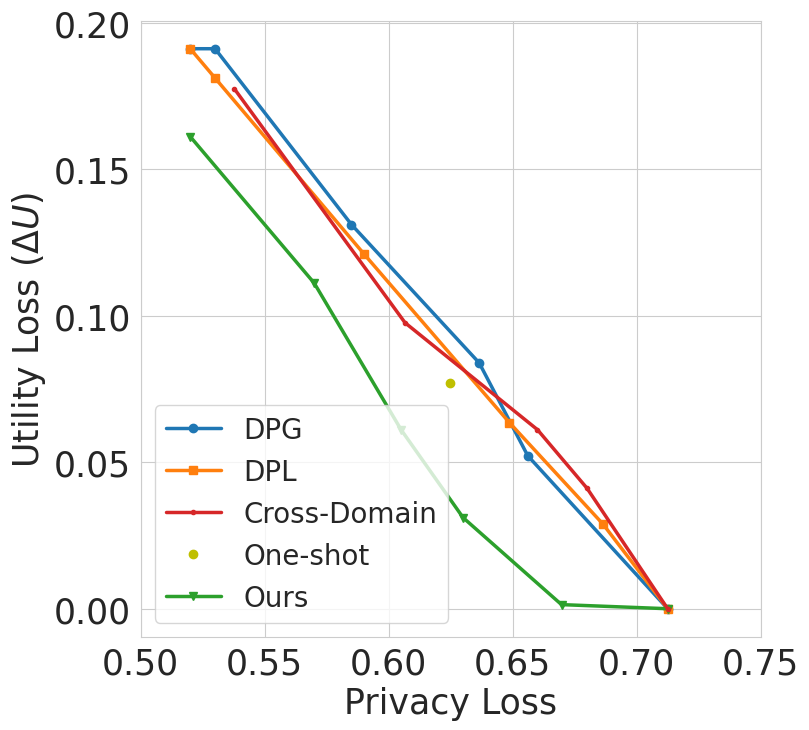}
}
\subfigure[Open-world-TI]{
\includegraphics[width=0.31\linewidth]{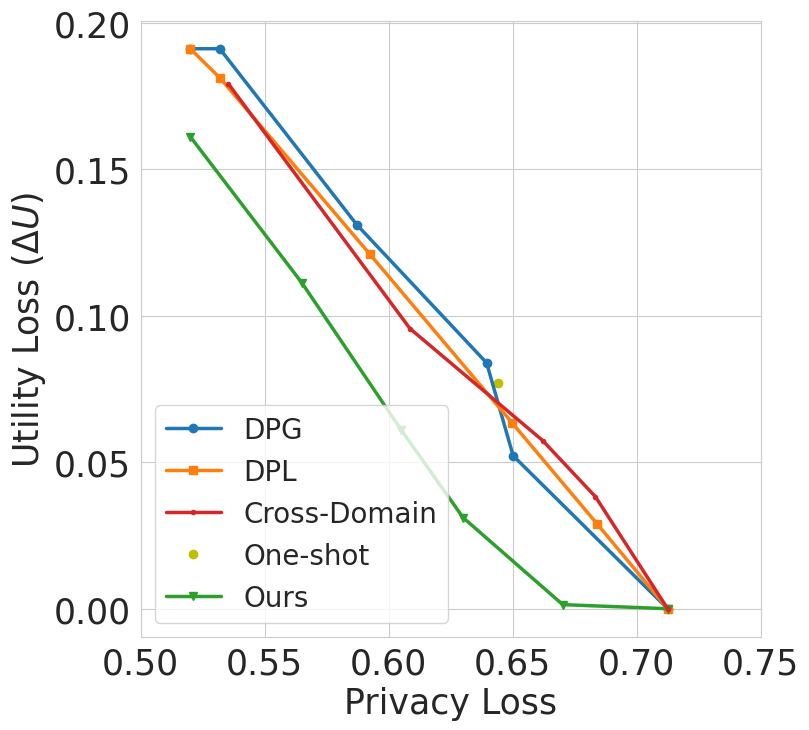}
}
\caption{Defense performance evaluation in Close-world,
Open-world-CF and Open-world-TI settings. The ideal trade-off curve resides on the bottom left corner in the figure.} \label{fig4}
\end{figure*}
The results of the state-of-the-art defense strategies in logit-based FL and our proposed FedLP (denoted as Ours in the figure) on three settings are reported in Figure \ref{fig4}. The X-axis represents the privacy loss, i.e. the capability for the adversary to infer a client’s private model. Y-axis represents the utility loss brought by the defense methods. Comparing to the state-of-the-art baselines, our proposed FedLP defense is closest to the ideal trade-off, which should reside in the bottom left corner in Figure \ref{fig4}. For example, when privacy loss is 0.7, the utility loss of our defense is around 8\% less than DP-G and DP-L in Close-world setting. This result indicates that our defense can provide a better utility and privacy trade-off compared to the state-of-the-art defense baselines.

\textbf{Effect of Hyperparameter (a,b)} We vary the values of the hyperparameters $(a,b)$ in the defense strategy under the Open-world-CF and Open-world-TI settings to evaluate their effect on defense performance. The results are presented in Table \ref{table8}. For clearer comparison, we fix the privacy loss at approximately 0.61 and compare the corresponding utility loss. The results indicate that FedLP demonstrates robustness to the hyperparameters $(a,b)$ in our settings.

\begin{table*}[t]
    \centering
    \caption{The effect of hyperparameter (a,b) on the defense performance of FedLP in Open-world-CF and Open-world-TI setting.}
    \label{table8}
    \begin{tabular}{ccccc}
        \toprule
       \multirow{2}{*}{Defense}  &  \multicolumn{2}{c}{Open-world-CF} & \multicolumn{2}{c}{Open-world-TI} \\
       \cline{2-3} \cline{4-5} 
       & Privacy Loss & Utility Loss(\%) & Privacy Loss & Utility Loss(\%) \\
       \hline
        (2,1) & 0.61 & 12.11 $\pm$ 2.44 & 0.61 & 11.88 $\pm$ 2.12 \\
        (2,2) & 0.61 & 12.05 $\pm$ 2.50 & 0.61 & 11.93 $\pm$ 2.49 \\
        \bottomrule
    \end{tabular}
\end{table*}

\section{Conclusion}
In this paper, we provide the first theoretical and empirical analysis of the privacy risk in logit-based FL that the semi-honest server can infer clients’ private models according to logits. To quantify the impacts of the privacy risk, we introduce AdaMSA, a method leveraging historical logits during training to enhance attack performance, and provide a theoretical analysis of the associated privacy risk bound. Moreover, we propose a perturbation-based defense named FedLP that perturbs the transmitted logit in the direction that minimizes the privacy risk while maximally preserving the training performance. We further provide theoretical analyses of its privacy and utility properties. Our experimental results showcase the effectiveness of AdaMSA and FedLP in various scenarios.

\bibliographystyle{ACM-Reference-Format}
\bibliography{sample-base}

@String{Computing = "Computing" }

@inproceedings{
allen-zhu2023towards,
title={Towards Understanding Ensemble, Knowledge Distillation and Self-Distillation in Deep Learning},
author={Zeyuan Allen-Zhu and Yuanzhi Li},
booktitle={The Eleventh International Conference on Learning Representations },
year={2023},
url={https://openreview.net/forum?id=Uuf2q9TfXGA}
}

@article{lecun1989handwritten,
  title={Handwritten digit recognition with a back-propagation network},
  author={LeCun, Yann and Boser, Bernhard and Denker, John and Henderson, Donnie and Howard, Richard and Hubbard, Wayne and Jackel, Lawrence},
  journal={Advances in neural information processing systems},
  volume={2},
  year={1989}
}

@INPROCEEDINGS{7958568,
  author={Shokri, Reza and Stronati, Marco and Song, Congzheng and Shmatikov, Vitaly},
  booktitle={2017 IEEE Symposium on Security and Privacy (SP)}, 
  title={Membership Inference Attacks Against Machine Learning Models}, 
  year={2017},
  volume={},
  number={},
  pages={3-18},
  doi={10.1109/SP.2017.41}}

@article{hsu2019measuring,
  title={Measuring the effects of non-identical data distribution for federated visual classification},
  author={Hsu, Tzu-Ming Harry and Qi, Hang and Brown, Matthew},
  journal={arXiv preprint arXiv:1909.06335},
  year={2019}
}

@inproceedings{takahashi2023breaching,
  title={Breaching FedMD: Image Recovery via Paired-Logits Inversion Attack},
  author={Takahashi, Hideaki and Liu, Jingjing and Liu, Yang},
  booktitle={CVPR},
  pages={12198--12207},
  year={2023}
}

@article{oliinyk2020first,
  title={First-Order Optimization (Training) Algorithms in Deep Learning},
  author={Oliinyk, Kyrylo},
  year={2020}
}

@article{crammer2008learning,
  title={Learning from Multiple Sources.},
  author={Crammer, Koby and Kearns, Michael and Wortman, Jennifer},
  journal={JMLR},
  volume={9},
  number={8},
  year={2008}
}

@article{allen2020towards,
  title={Towards understanding ensemble, knowledge distillation and self-distillation in deep learning},
  author={Allen-Zhu, Zeyuan and Li, Yuanzhi},
  journal={arXiv preprint arXiv:2012.09816},
  year={2020}
}

@inproceedings{cohen2017emnist,
  title={EMNIST: Extending MNIST to handwritten letters},
  author={Cohen, Gregory and Afshar, Saeed and Tapson, Jonathan and Van Schaik, Andre},
  booktitle={IJCNN},
  pages={2921--2926},
  year={2017},
  organization={IEEE}
}

@article{krizhevsky2017imagenet,
  title={Imagenet classification with deep convolutional neural networks},
  author={Krizhevsky, Alex and Sutskever, Ilya and Hinton, Geoffrey E},
  journal={Communications of the ACM},
  volume={60},
  number={6},
  pages={84--90},
  year={2017},
  publisher={AcM New York, NY, USA}
}

@article{le2015tiny,
  title={Tiny imagenet visual recognition challenge},
  author={Le, Ya and Yang, Xuan},
  journal={CS 231N},
  volume={7},
  number={7},
  pages={3},
  year={2015}
}

@inproceedings{papernot2017practical,
  title={Practical black-box attacks against machine learning},
  author={Papernot, Nicolas and McDaniel, Patrick and Goodfellow, Ian and Jha, Somesh and Celik, Z Berkay and Swami, Ananthram},
  booktitle={ACM ASIACCS},
  pages={506--519},
  year={2017}
}

@inproceedings{orekondy2019knockoff,
  title={Knockoff nets: Stealing functionality of black-box models},
  author={Orekondy, Tribhuvanesh and Schiele, Bernt and Fritz, Mario},
  booktitle={CVPR},
  pages={4954--4963},
  year={2019}
}

@inproceedings{tramer2016stealing,
  title={Stealing Machine Learning Models via Prediction APIs.},
  author={Tram{\`e}r, Florian and Zhang, Fan and Juels, Ari and Reiter, Michael K and Ristenpart, Thomas},
  booktitle={USENIX security symposium},
  volume={16},
  pages={601--618},
  year={2016}
}

@article{blitzer2007learning,
  title={Learning bounds for domain adaptation},
  author={Blitzer, John and Crammer, Koby and Kulesza, Alex and Pereira, Fernando and Wortman, Jennifer},
  journal={NeurIPS},
  volume={20},
  year={2007}
}

@InProceedings{karimireddy2019scaffold,
  title = 	 {{SCAFFOLD}: Stochastic Controlled Averaging for Federated Learning},
  author =       {Karimireddy, Sai Praneeth and Kale, Satyen and Mohri, Mehryar and Reddi, Sashank and Stich, Sebastian and Suresh, Ananda Theertha},
  booktitle = 	 {ICML},
  pages = 	 {5132--5143},
  year = 	 {2020},
  editor = 	 {},
  volume = 	 {119},
  series = 	 {PMLR},
  month = 	 {13--18 Jul},
  publisher =    {PMLR},
  url = 	 {https://proceedings.mlr.press/v119/karimireddy20a.html}
}

@article{jeong2018communication,
  title={Communication-efficient on-device machine learning: Federated distillation and augmentation under non-iid private data},
  author={Jeong, Eunjeong and Oh, Seungeun and Kim, Hyesung and Park, Jihong and Bennis, Mehdi and Kim, Seong-Lyun},
  journal={arXiv preprint arXiv:1811.11479},
  year={2018}
}

@inproceedings{gong2021ensemble,
  title={Ensemble attention distillation for privacy-preserving federated learning},
  author={Gong, Xuan and others},
  booktitle={ICCV},
  pages={15076--15086},
  year={2021}
}

@article{netzer2011reading,
  title={Reading digits in natural images with unsupervised feature learning.},
  author={Netzer, Yuval and Wang, Tao and Coates, Adam and Bissacco, Alessandro and Wu, Bo and Ng, Andrew Y},
  year={2011}
}

@article{lecun1998gradient,
  title={Gradient-based learning applied to document recognition},
  author={LeCun, Yann and Bottou, L{\'e}on and Bengio, Yoshua and Haffner, Patrick},
  journal={Proceedings of the IEEE},
  volume={86},
  number={11},
  pages={2278--2324},
  year={1998},
  publisher={Ieee}
}

@article{krizhevsky2009learning,
  title={Learning multiple layers of features from tiny images},
  author={Krizhevsky, Alex and others},
  year={2009},
  publisher={Toronto, ON, Canada}
}

@article{sattler2021fedaux,
  title={Fedaux: Leveraging unlabeled auxiliary data in federated learning},
  author={Sattler, Felix and Korjakow, Tim and Rischke, Roman and Samek, Wojciech},
  journal={IEEE TNNLS},
  year={2021},
  publisher={IEEE}
}

@article{lin2020ensemble,
  title={Ensemble distillation for robust model fusion in federated learning},
  author={Lin, Tao and others},
  journal={NeurIPS},
  volume={33},
  pages={2351--2363},
  year={2020}
}

@inproceedings{gong2022preserving,
  title={Preserving privacy in federated learning with ensemble cross-domain knowledge distillation},
  author={Gong, Xuan and others},
  booktitle={AAAI},
  volume={36},
  number={11},
  pages={11891--11899},
  year={2022}
}

@article{geiping2020inverting,
  title={Inverting gradients-how easy is it to break privacy in federated learning?},
  author={Geiping, Jonas and Bauermeister, Hartmut and Dr{\"o}ge, Hannah and Moeller, Michael},
  journal={NeurIPS},
  volume={33},
  pages={16937--16947},
  year={2020}
}

@article{zhu2019deep,
  title={Deep leakage from gradients},
  author={Zhu, Ligeng and Liu, Zhijian and Han, Song},
  journal={NeurIPS},
  volume={32},
  year={2019}
}

@article{kairouz2019advances,
  title={Advances and open problems in federated learning},
  author={Kairouz, Peter and McMahan, H Brendan and Avent, Brendan and Bellet, Aur{\'e}lien and Bennis, Mehdi and Bhagoji, Arjun Nitin and Bonawitz, Kallista and Charles, Zachary and others},
  journal={arXiv preprint arXiv:1912.04977},
  year={2019}
}

@inproceedings{wang2019beyond,
  title={Beyond inferring class representatives: User-level privacy leakage from federated learning},
  author={Wang, Zhibo and others},
  booktitle={IEEE INFOCOM},
  pages={2512--2520},
  year={2019},
  organization={IEEE}
}

@inproceedings{nasr2019comprehensive,
  title={Comprehensive privacy analysis of deep learning: Passive and active white-box inference attacks against centralized and federated learning},
  author={Nasr, Milad and Shokri, Reza and Houmansadr, Amir},
  booktitle={2019 IEEE symposium on security and privacy (SP)},
  pages={739--753},
  year={2019},
  organization={IEEE}
}

@inproceedings{melis2019exploiting,
  title={Exploiting unintended feature leakage in collaborative learning},
  author={Melis, Luca and Song, Congzheng and De Cristofaro, Emiliano and Shmatikov, Vitaly},
  booktitle={2019 IEEE symposium on security and privacy (SP)},
  pages={691--706},
  year={2019},
  organization={IEEE}
}

@article{li2019fedmd,
  title={Fedmd: Heterogenous federated learning via model distillation},
  author={Li, Daliang and Wang, Junpu},
  journal={arXiv preprint arXiv:1910.03581},
  year={2019}
}

@article{hinton2015distilling,
  title={Distilling the knowledge in a neural network},
  author={Hinton, Geoffrey and Vinyals, Oriol and Dean, Jeff},
  journal={arXiv preprint arXiv:1503.02531},
  year={2015}
}

@inproceedings{mcmahan2016communication,
  title={Communication-efficient learning of deep networks from decentralized data},
  author={McMahan, Brendan and Moore, Eider and Ramage, Daniel and Hampson, Seth and y Arcas, Blaise Aguera},
  booktitle={Artificial intelligence and statistics},
  pages={1273--1282},
  year={2017},
  organization={PMLR}
}

@article{li2020practical,
  title={Practical one-shot federated learning for cross-silo setting},
  author={Li, Qinbin and He, Bingsheng and Song, Dawn},
  journal={arXiv preprint arXiv:2010.01017},
  volume={1},
  number={3},
  year={2020}
}

@ArtifactSoftware{R,
    title = {R: A Language and Environment for Statistical Computing},
    author = {{R Core Team}},
    organization = {R Foundation for Statistical Computing},
    address = {Vienna, Austria},
    year = {2019},
    url = {https://www.R-project.org/},
}

\newpage
\appendix

\begin{center}
{\Large\bfseries Supplementary Material for}\\[0.5em]
{\large\bfseries Quantifying and Defending against the Privacy Risk in Logit-based Federated Learning}
\end{center}

\section{Proofs of Theoretical Results for FedLP}
\label{appendix}

\subsection{Proof of Theorem~\ref{thm:gradient_sensitivity}}
\begin{proof}[Proof of Theorem~\ref{thm:gradient_sensitivity}]
For cross-entropy with softmax targets, the gradient with respect to the model parameters satisfies
\begin{equation}
G_i(\theta,p_i)
=
J_{\theta}(x_i)^\top
\left(s(f(x_i,\theta))-s(p_i)\right).
\end{equation}
Hence,
\begin{equation}
\begin{aligned}
&G_i(\theta_{T-1},p_i+Ze_j)
-G_i(\theta_{T-1},p_i)\\
&\quad=
J_{\theta_{T-1}}(x_i)^\top
\left(s(p_i)-s(p_i+Ze_j)\right).
\end{aligned}
\end{equation}

Using the first-order Taylor expansion of softmax,
\begin{equation}
s(p_i+Ze_j)
=
s(p_i)
+
Z\left(
\operatorname{diag}(s(p_i))
-s(p_i)s(p_i)^\top
\right)e_j
+O(Z^2).
\end{equation}
Substituting this expression yields
\begin{equation}
G_i(\theta_{T-1},p_i+Ze_j)
-G_i(\theta_{T-1},p_i)
=
-ZA_ie_j+O(Z^2).
\end{equation}
Taking the Euclidean norm gives
\begin{equation}
\left\|
G_i(\theta_{T-1},p_i+Ze_j)
-G_i(\theta_{T-1},p_i)
\right\|_2
=
|Z|\|A_ie_j\|_2+O(Z^2).
\end{equation}
Summing over the finite public dataset completes the proof.
\end{proof}

\subsection{Proof of Theorem~\ref{thm:directional_advantage}}
\begin{proof}[Proof of Theorem~\ref{thm:directional_advantage}]
Define
\begin{equation}
a_j=
\sum_{x_i\in\mathcal{D}_{pub}}\|A_ie_j\|_2.
\end{equation}
By Theorem~\ref{thm:gradient_sensitivity},
\begin{equation}
D_j(Z)=|Z|a_j+O(Z^2).
\end{equation}
Since FedLP selects the direction maximizing $D_j(Z)$, and the softmax derivatives are bounded for the fixed models and finite public dataset,
\begin{equation}
\mathbb{E}[D_{\mathrm{FedLP}}]
=
\sigma\sqrt{\frac{2}{\pi}}
\max_j a_j
+O(\sigma^2).
\end{equation}

For the Gaussian baseline, the first-order gradient discrepancy is
\begin{equation}
G_i(\theta_{T-1},p_i+\xi)
-G_i(\theta_{T-1},p_i)
=
-A_i\xi+O(\|\xi\|_2^2).
\end{equation}
By the triangle inequality,
\begin{equation}
\begin{aligned}
\mathbb{E}\|A_i\xi\|_2
&=
\mathbb{E}
\left\|
\sum_{j=1}^{C}\xi_j A_ie_j
\right\|_2\\
&\leq
\sum_{j=1}^{C}
\mathbb{E}|\xi_j|\|A_ie_j\|_2\\
&=
\frac{\sigma}{C}
\sqrt{\frac{2}{\pi}}
\sum_{j=1}^{C}\|A_ie_j\|_2.
\end{aligned}
\end{equation}
Since $\mathbb{E}\|\xi\|_2^2=\sigma^2/C$, summing over the public dataset gives
\begin{equation}
\mathbb{E}[D_{\mathrm{G}}]
\leq
\frac{\sigma}{C}\sqrt{\frac{2}{\pi}}
\sum_{j=1}^{C}a_j
+O(\sigma^2).
\end{equation}

Combining the two results,
\begin{equation}
\begin{aligned}
&\mathbb{E}[D_{\mathrm{FedLP}}]
-\mathbb{E}[D_{\mathrm{G}}]\\
&\quad\geq
\sigma\sqrt{\frac{2}{\pi}}
\left(
\max_j a_j-\frac{1}{C}\sum_{j=1}^{C}a_j
\right)
-O(\sigma^2).
\end{aligned}
\end{equation}

Dividing by $\sigma$ and taking the lower limit as $\sigma\to 0$ yields the stated bound, since
\begin{equation}
\max_j a_j
\geq
\frac{1}{C}\sum_{j=1}^{C}a_j.
\end{equation}
This completes the proof.
\end{proof}

\subsection{Proof of Theorem~\ref{thm:utility_distortion}}
\begin{proof}[Proof of Theorem~\ref{thm:utility_distortion}]
Let $\delta=p'-p$ and define the log-sum-exp function
\begin{equation}
F(p)=\log\left(\sum_{j=1}^{C}e^{p_j}\right).
\end{equation}
Its gradient and Hessian are
\begin{equation}
\nabla F(p)=s(p),
\end{equation}
and
\begin{equation}
\nabla^2 F(p)
=
\operatorname{diag}(s(p))-s(p)s(p)^\top.
\end{equation}

For any unit vector $v$, we have
\begin{equation}
v^\top\nabla^2 F(p)v
=
\sum_{j=1}^{C}s_j(p)v_j^2
-
\left(\sum_{j=1}^{C}s_j(p)v_j\right)^2.
\end{equation}
This expression is the variance of a random variable taking value $v_j$ with probability $s_j(p)$. By the variance bound and
\[
(\max_j v_j-\min_j v_j)^2
\leq 2\|v\|_2^2=2,
\]
we obtain
\begin{equation}
0\preceq\nabla^2F(p)\preceq\frac12 I_C.
\end{equation}

The KL divergence between the original and perturbed prediction distributions can be expressed as
\begin{equation}
\begin{aligned}
D_{\mathrm{KL}}(s(p)\|s(p'))
&=F(p+\delta)-F(p)\\
&\quad-\nabla F(p)^\top\delta.
\end{aligned}
\end{equation}

Applying the second-order Taylor formula in integral form,
\begin{equation}
\begin{aligned}
&D_{\mathrm{KL}}(s(p)\|s(p'))\\
&=
\int_0^1(1-t)
\delta^\top\nabla^2F(p+t\delta)\delta\,dt\\
&\leq
\frac12\|\delta\|_2^2
\int_0^1(1-t)\,dt\\
&=\frac14\|\delta\|_2^2.
\end{aligned}
\end{equation}

Since $\|\delta\|_2\leq\|\delta\|_1$, it follows that
\begin{equation}
D_{\mathrm{KL}}(s(p)\|s(p'))
\leq\frac14\|\delta\|_1^2.
\end{equation}
When $\|\delta\|_1\leq\gamma$, the deterministic bound follows immediately.

For FedLP, $\delta=Ze_{j^*}$, where $\|e_{j^*}\|_2=1$. Thus,
\begin{equation}
D_{\mathrm{KL}}(s(p)\|s(p'))
\leq\frac{Z^2}{4}.
\end{equation}
Taking expectations over $Z\sim\mathcal{N}(0,\sigma^2)$ yields
\begin{equation}
\begin{aligned}
\mathbb{E}_Z
[D_{\mathrm{KL}}(s(p)\|s(p'))]
&\leq\frac14\mathbb{E}_Z[Z^2]\\
&=\frac{\sigma^2}{4}.
\end{aligned}
\end{equation}
This bound holds even when the selected direction $j^*$ depends on $Z$, since $\|Ze_{j^*}\|_2^2=Z^2$ for every realization. This completes the proof.
\end{proof}

\end{document}